\documentclass[reqno,12pt]{amsart}

\usepackage{amssymb}
\usepackage{amsfonts}
\usepackage{amsmath}
\usepackage{amsthm}
\usepackage{booktabs}
\usepackage{graphicx}
\usepackage{xcolor}
\usepackage[all]{xy}
\usepackage{hyperref}

\newcommand\Z{\mathbb{Z}}
\newcommand\Q{\mathbb{Q}}
\newcommand\R{\mathbb{R}}
\newcommand\C{\mathbb{C}}

\newcommand{\calC}{{\mathcal C}}

\newcommand\Hom{\operatorname{Hom}}

\newcommand\sign{\operatorname{sign}}
\newcommand{\sig}{\operatorname{sig}}
\newcommand{\odeg}{\operatorname{odeg}}
\newcommand{\ord}{\operatorname{ord}}

\newtheorem{definition}{Definition}

\newtheorem{theorem}{Theorem}
\newtheorem{proposition}{Proposition}

\newtheorem{example}{Example}
\newtheorem{remark}{Remark}

\begin{document}

\title[Simplicial degree in complex networks]{Simplicial degree in complex networks. Applications of topological data analysis to Network Science.}

\date{\today}
\thanks{This work has been supported by Ministerio de Econom\'ia y Competitividad (Spain) and the European Union through FEDER funds under grants TIN2017-84844-C2-1-R, MTM2017-86042-P and TEC2015-68734-R, the project STAMGAD 18.J445 / 463AC03 supported by Consejer\'{\i}a de Educaci\'on (GIR, Junta de Castilla y Le\'on, Spain), and the Generalitat de Catalunya grant 2017-SGR-00782.
}



\author[D. Hern\'andez Serrano]{Daniel Hern\'andez Serrano}
\author[J. Hern\'andez Serrano]{Juan Hern\'andez Serrano}
\author[D. S\'anchez G\'omez]{Dar\'io S\'anchez G\'omez}

\address{Departamento de Matem\'aticas and Instituto Universitario de F\'isica Fundamental y Matem\'aticas (IUFFyM), Universidad de Salamanca, Salamanca, Spain}
\email{dani@usal.es, dario@usal.es}

\address{Department of Network Engineering, Universitat Polit\`ecnica de Catalunya, Barcelona, Spain.}
\email{j.hernandez@upc.edu}

\begin{abstract}

Network Science provides a universal formalism for modelling and studying complex systems based on pairwise interactions between agents. However, many real networks in the social, biological or computer sciences involve interactions among more than two agents, having thus an inherent structure of a simplicial complex. The relevance of an agent in a graph network is given in terms of its degree, and in a simplicial network there are already notions of adjacency and degree for simplices that, as far as we know, are not valid for comparing simplices in different dimensions. We propose new notions of higher-order degrees of adjacency for simplices in a simplicial complex, allowing any dimensional comparison among them and their faces. We introduce multi-parameter boundary and coboundary operators in an oriented simplicial complex and also a novel multi-combinatorial Laplacian is defined. As for the graph or combinatorial Laplacian, the multi-combinatorial Laplacian is shown to be an effective tool for calculating the higher-order degrees presented here. To illustrate the potential applications of these theoretical results, we perform a structural analysis of higher-order connectivity in simplicial-complex networks by studying the associated distributions with these simplicial degrees in 17 real-world datasets coming from different domains such as coauthor networks, cosponsoring Congress bills, contacts in schools, drug abuse warning networks, e-mail networks or publications and users in online forums. We find rich and diverse higher-order connectivity structures and observe that datasets of the same type reflect similar higher-order collaboration patterns. Furthermore, we show that if we use what we have called the maximal simplicial degree (which counts the distinct maximal communities in which our simplex and all its strict sub-communities are contained), then its degree distribution is, in general, surprisingly different from the classical node degree distribution.
\end{abstract}

\maketitle
{\small \tableofcontents}

\section{Introduction} \label{intro}

The simplest way to mathematically describe real-world networks is to use graphs, where nodes represent the agents of the network and the edges are thought to be the interactions between these agents. These systems are called complex networks in the literature and Network Science and Statistical Mechanics of complex networks (\cite{BR02, BA16}) provide a universal language which allows to classify networks, elucidate patterns of interactions and make predictions about the structure and evolution of such systems.

Despite the success of complex networks analysis and Network Science, there is a major drawback to this approach due to the fact that an implicit assumption is made: the complex system is described by combinations of pairwise interactions, that is, binary relations. Nonetheless, many complex systems and datasets in real-world networks come with a richer inherent structure, since there are higher-order interactions involving group of agents. Thus, studying higher-order structural connectivity is essential to understanding complex networks. Many of the most successful results with regard to this are based on the use of a powerful algebraic-topology tool used in Topological Data Analysis (TDA) known as simplicial complexes. Simplicial complexes generalize the standard graph tools by allowing many-body interactions, providing robust results under continuous deformations of the system or dataset. The basic idea is simple, in graph theory we cannot distinguish the three agents which are pairwise linked (for example, if they have written a paper pairwise, so that it is represented as a triangle) from the situation where the three of them have published a joint paper (and thus, in particular they have written it also pairwise, again a triangle but this time is filled). In simplicial theory, the filled triangle is a $2$-simplex which, by definition, contains all of its faces ($1$-simplices are pairwise connections and $0$-simplices are the agents). This simplicial point of view naturally allows to keep track of the multi-interactions among the agents or group of agents. Given the finite set of vertices $\{v_0,v_1,\dots ,v_n\}$ in a network, a $q$-simplex is a subset $\sigma^{(q)}=\{v_0,v_1,\dots ,v_q\}$ such that $v_i\neq v_j$ for all $i\neq j$, and a $p$-face (for $p\leq q$) of $\sigma^{(q)}$ is just a subset $\tau^{(p)}=\{v_{i_0},\dots, v_{i_p}\}$ of $\sigma^{(q)}$. A simplicial complex $K$ is a collection of simplices such that if $\sigma$ is a simplex in $K$, then all the faces of $\sigma$ are also in $K$. The number $q$ is called the dimension of the simplex $\sigma^{(q)}$.

The application of TDA on the study of higher-order structures in complex systems has produced significant advances in a wide variety of domains; see for instance \cite{BASJK18,BK19,CH13,MDS15,MR12,MR14,PGV17} for social systems, \cite{ER18,GH15,KFH16,XW14,XW15} for biomolecular systems, \cite{CL18,GGB15,GPCI15,LC19,P14,SPGB18} for brain networks, \cite{ME06} for network control and sensing, \cite{HN15} for material science or \cite{Ghrist08} for a survey on algebraic topology tools and data and a novel use of barcodes in natural images. Many of the above contributions make use of the homology and persistent homology of a simplicial complex to analyse topological aspects of the data collected. But in order to study higher-order patterns of connectivity in a complex network,  we will focus on a different aspect: the relevance or influence of a simplicial community in terms of a new generalized simplicial degree, both from a theoretical and an applied point of view.

This influence is measured in a graph network by the degree of a node, the number of its incident edges (or upper adjacent ones), but when dealing with simplicial complexes, lower adjacency might appear (think of a triangle where its three edges are lower adjacent to it, or a triangle having another triangle attached to it on one of its edges) and thus different definitions of adjacency are required. As far as we know, the notions of adjacency for simplices used in the literature are the following: two simplices $\sigma^{(q)}$ and $\sigma^{(q)}$ are lower adjacent if there exists a $(q-1)$-simplex $\tau^{(q-1)}$ which is a common $(q-1)$-face of both of them; they are said to be upper adjacent if there exists $\tau^{(q+1)}$ having both as $q$-faces; and they are considered adjacent if they are strictly lower adjacent but not upper adjacent. Thus, it is possible to upper compare two $q$-simplices if they are faces of the same (one more dimensional) $(q+1)$-simplex, and lower compare them if they share a common (one less dimensional) $(q-1)$-simplex. See for instance \cite{G02, MR12, MR14} for theoretical or social-networks implications, \cite{ER18} for protein interactions or \cite{HJ13, MS13,PR17} for spectral theory and random walks on simplicial complexes. The notion of upper adjacency is also further expanded in \cite{MDS15} following the idea that a vertex-to-triangle degree can be computed by counting the number of triangles incident to each edge, which is incident to the vertex, and then dividing by $2$. They propose a $q$-simplex to $(q+h)$-simplex upper degree of a $q$-simplex and a vertex-to-facet upper degree (a facet of a simplicial complex is a maximal simplex under the set inclusion), and affirm that, given a facet list, this degree can be ``computed with a relatively straightforward searching and counting procedure'', but no explicit theoretical or computational method is given.

None of the above definitions allow us to compare simplices of different dimensions sharing (or being contained in) any dimensional face. Additionally, the upper simplicial degree of \cite{MDS15} does not allow us to count the collaborations of the strict faces of a simplex with other simplicial communities, since it does not take into account lower adjacency or general adjacency. These are the reasons why we are introducing in this paper a mathematical framework for generalizing the notions of lower, upper and general adjacency and their associated degrees, which are valid for any simplicial dimension comparisons. We define new notions of higher-order upper and lower adjacency that generalize all the notions commented above. We show how to explicitly compute the strict lower and upper degrees by giving a closed formula, which in particular states an explicit mechanism to compute the ``$q$-simplex to facets degree'' of \cite{MDS15} as a type of upper simplicial degree. We will also give closed formulas for our general adjacency degree.

On the other hand, it is known that the degree of a node and the binary adjacency relations between nodes in a graph network are codified in the graph Laplacian matrix. This result also holds in the simplicial case (see for example \cite{G02}) by using the $q$-combinatorial Laplacian (\cite{E44}) and the definitions of simplicial degree mentioned above (comparing two simplices by one more or one less dimensional face). For the notion of $q$-simplex to $(q+h)$-simplex upper degree of \cite{MDS15}, the counting procedure proposed can be also stated in terms of certain entries of the product several combinatorial Laplacian matrices, but not from a single Laplacian one. Therefore, and in the same spirit, we will define a higher-order multi-parameter boundary operator which allows us to introduce a novel higher-order multi-combinatorial Laplacian, generalizing both the graph Laplacian and the $q$-combinatorial Laplacian. The entries of the associated multi-combinatorial Laplacian matrix compute some of the higher-order simplicial degrees here defined. In addition, we use the boundary and coboundary operators to state closed and effective formulas for all the higher-order simplicial degrees in a simplicial network.

From an applied point of view, we propose two notions as important for studying the degree of relevance of a simplicial community: the maximal upper simplicial degree of a simplex (counting the distinct maximal simplicial communities on which the simplex is nested in), and the maximal simplicial degree of a simplex (counting the distinct maximal simplicial communities on which the simplex is nested in and also the different maximal communities the strict faces of the simplex are contained in). We perform a structural analysis of the higher-order connectivity of 17 real-world datasets given in \cite{BASJK18}. These datasets are collected as simplices with their corresponding nodes (bounded to a maximum of 25 nodes), and are obtained from real-world data from a rich variety of  domains such as coauthor networks, cosponsoring Congress bills, contacts in schools, drug abuse warning networks, e-mail networks, national drug code classes and substances, publications in online forums or users in online forums. We study some of their statistical properties and the degree distributions associated with the simplicial degrees proposed in these notes, and compare these distributions with the classical node degree and node-to-facets degree distributions of the datasets. We observe that while the collaboration sizes tend to be varied among the different categories, the distributions of the facets size are similar across most datasets. We show that when studying the distributions of our proposed simplicial degrees, rich and varied higher-order connectivity structures make apparent but nonetheless, and as it would be expected, datasets of the same type reflect similar higher-order patterns. Moreover, the same type of classical node and node-to-facets degree distributions are observed, something which is consistent with the results and observations of \cite{PGV17} and \cite{MDS15}. We prove that for every dataset the degree distribution associated with the maximal upper simplicial degree of Definition \ref{d:simpdeg} is closer to a power law distribution having a more pronounced decay. Furthermore, we show that if we use the maximal simplicial degree of Definition \ref{d:simpdeg}, which counts the facets on which a simplex is contained in and also counts the facets on which the strict sub communities of the simplex are contained in, then its degree distribution is, in general, surprisingly different from the classical and upper adjacency ones.

The rest of this paper is organized as follows. Section \ref{sec:simplicial} starts by recalling well-known definitions and properties in simplicial complexes. In Section \ref{s:qhdeg} we introduce the new notions of simplicial adjacencies and degrees (which generalize the usual notions) and we provide  closed formulas to explicitly compute the higher-order simplicial degrees. In Section \ref{s:qhLap}, a new higher-order multi-combinatorial Laplacian operator is defined by using a novel multi-parameter boundary operator and we explicitly compute the entries of the matrix associated with the  multi-combinatorial Laplacian in terms of the simplicial degrees. We present in Section \ref{s:realapp} an application of the theoretical results: a structural higher-order connectivity analysis in simplicial networks is done by studying statistical properties and simplicial degree distributions of several and varied real-world datasets. A brief summary of some of the results and a set possible lines for future research are given in Section \ref{s:concl}. Finally, we present in Appendix \ref{s:A} the proofs of some theoretical results concerning the multi-combinatorial Laplacian.

\section{Simplicial  complexes, adjacency and combinatorial Laplacian}\label{sec:simplicial}

Simplicial complexes have been very much studied in the literature and during the last decade they have been proved to be a powerful tool in Topological Data Analysis (TDA). Very recently, the simplicial techniques of TDA are being also applied in the context of Complex Networks and Network Science. We shall start with some well-known definitions and properties on the category of  simplicial complexes. We refer to  (\cite{Mun84,G02}) for a wide exposition and details.

Roughly speaking, given a finite set of points $\{v_0,v_1,\dots ,v_n\}$, which we call vertices, a $q$-simplex is a subset of vertices $\{v_0,v_1,\dots ,v_q\}$ such that $v_i\neq v_j$ for all $i\neq j$. A $p$-face (for $p<q$) of a $q$-simplex is just a subset $\{v_{i_1},\dots, v_{i_p}\}$ of the $q$-simplex. A simplicial complex $K$ is a collection of simplices such that if $\sigma$ is a simplex in $K$, then all the faces of $\sigma$ are also in $K$. 

Formally, a set $\{v_0,\dots, v_q\}$ of points of $\mathbb{R}^n$  is said to be geometrically independent if the vectors $\{v_0-v_1,\dots, v_0-v_q\}$ are linearly independent. 

The $q$-simplex spanned by these points is the convex envelope, that is, the set of all points of $\mathbb{R}^n$ such that
$$\sigma=\Big\{\sum_{i=0}^q \lambda_i v_i \,\colon\, \sum_{i=0}^q\lambda_i=1 \text{ and } \lambda_i\geq 0\text{ for all } i\Big\}$$
The points $\{v_0,\dots, v_q\}$ that span $\sigma$ are called vertices of $\sigma$ and the number $q$ is the dimension of $\sigma$. The simplex spanned by a proper nonempty subset $\{v_{i_1},\dots, v_{i_p}\}$ of $\{v_0,\dots, v_q\}$ is called a $p$-face of $\sigma$. If a simplex is not a face of any other simplex, then it is called a facet.

A (finite) simplicial complex in $\mathbb{R}^n$ is a (finite) collection $K$ of simplicies in $\mathbb{R}^n$ satisfying the following conditions:
\begin{enumerate}
\item If $\sigma\in K$ and $\tau$ is a face of $\sigma$, then $\tau\in K$.
\item The non-empty intersection of any two simplices of $K$ is a face of each of them.
\end{enumerate}
Each element $\sigma\in K$ is called a $q$-simplex of $K$, being $q+1$ the cardinality of $\sigma$. The union of $0$-simplices of $K$ is called the vertex set of $K$. The dimension of $K$ is defined as $\dim K= \max\{\dim \sigma \colon \sigma\in K\}$. We shall use the notation $\sigma^{(q)}$ to denote a simplex $\sigma$ of dimension $q$.

Hence, simplices can be understood as higher dimensional generalizations of a point, line, triangle, tetrahedron, and so on. Since simplices can codify multi interaction relations in classical networks (co-authorship network, social networks, protein interaction network, biological networks, ...), they are starting to be introduced in Network Science. 

\begin{remark}
Even if simplicial complexes and many of its associated properties can be defined over a commutative ring with unity, for the sake of clarity we shall restrict ourselves to the base field $\R$.
\end{remark}

Recall that the degree of a vertex is the number of its incident edges. It has local relevance in determining the centrality of a vertex and global importance in modelling the network in virtue of its degree distribution. This notion can be generalized to q-simplices. 

Notice that as a $0$-simplex, a vertex has degree $d$ if there are $d$ edges, $1$-simplices, incident to it, but a $1$-simplex have two $0$-simplices adjacent to it (the two vertices the edge has) but it also might be adjacent to a $2$-simplex (triangle). That is, we need a notion of upper and lower adjacency in order to define the degree for $q$-simplices when $q>0$ (see \cite{ER18,G02} for details).

\begin{definition}\label{d:qULAdj}
Two $q$-simplices $\sigma_i^{(q)}$ and $\sigma_j^{(q)}$ of a simplicial complex $K$ are lower adjacent if they share a common $(q-1)$-face, which is called their common lower simplex. Lower adjacency is denoted as $\sigma_i^{(q)} \sim_L \sigma_j^{(q)}$.
 
Two $q$-simplices $\sigma_i^{(q)}$ and $\sigma_j^{(q)}$ of a simplicial complex $K$ are upper adjacent if they are both faces of the same common $(q+1)$-simplex, called their common upper simplex. Upper adjacency is denoted as $\sigma_i^{(q)} \sim_U \sigma_j^{(q)}$. 
\end{definition}
Notice that if two $q$-simplices are upper adjacent, then they are also lower adjacent. Moreover, if $\sigma_i^{(q)}$ and $\sigma_j^{(q)}$ are upper adjacent (resp. lower adjacent), then their common upper $(q+1)$-simplex (resp. their common lower simplex) is unique.

\begin{definition}\label{d:qLUdeg}
The lower degree of a $q$-simplex $\sigma^{(q)}$, denoted $\deg_L(\sigma^{(q)})$, is the number of $(q-1)$-simplices in $K$ which are contained in $\sigma^{(q)}$, which is always $\binom{q+1}{q}=q+1$. The upper degree of a $q$-simplex $\sigma^{(q)}$, denoted $\deg_U(\sigma)$, is the number of $(q+1)$-simplices in $K$ of which $\sigma^{(q)}$ is a $q$-face.

The degree of a $q$-simplex is defined as:
$$\deg(\sigma^{(q)}):=\deg_L(\sigma^{(q)})+\deg_U(\sigma^{(q)})=\deg_U(\sigma^{(q)})+q+1\,.$$
\end{definition}

In network theory, the degree of a vertex also appears as a diagonal entry of the graph Laplacian matrix, defined as $D-A$, where $D$ is a diagonal matrix with the degree of the vertices as diagonal entries, and $A$ is the usual vertex adjacency matrix. We will recall here the definition of the $q$-combinatorial Laplacian for $q$-simplices, and that of its associated matrix (which takes control of the degrees of $q$-simplices in a simplicial complex $K$ and their adjacency relations). As we shall see, the $q$-Laplacian operator makes use of the $q$-boundary operator, so that, an orientation is needed in the simplicial complex.

Let $\sigma$ be a simplex, we define two orderings of its vertex set to be equivalent if they differ from one another by an even permutation. If $\dim\sigma >0$, this relations provides two equivalence classes and each of them is called an orientation of $\sigma$. An oriented simplex is a simplex $\sigma$ together with an orientation of $\sigma$. For a geometrically independent set of points $\{v_0, v_1,\dots, v_q\}$ we denote by $[v_0,\dots, v_q]$ and $-[v_0,\dots, v_q]$ the opposite oriented simplices spanned by $\{v_0, v_1,\dots, v_q\}$. We say that a finite simplicial complex $K$ is oriented if all of its simplices are oriented. 
We shall denote by $\widetilde{S}_p(K)$ the set of oriented $p$-simplices of the simplicial complex $K$, and by $S_p(K)$ the set of non oriented $p$-simplices. 
 
Given and oriented simplicial complex $K$, we define the  group of $q$-chains as the free abelian group $C_q(K)$ with basis the set of oriented $q$-simplices of $K$. The dimension $f_q$ of $C_q(K)$ is the number of $q$-dimensional simplices of the simplicial complex $K$, and it is codified in a topological invariant called the $f$-vector $f=(f_0,f_1,\dots ,f_q, \dots ,f_{\dim K})$. By assumption $C_q(K)$ is trivial if $q\notin[0,\dim K]$.

\begin{definition}
The $q$-boundary operator $\partial_q\colon C_q(K)\to C_{q-1}(K)$ is the homomorphism given as the linear extension of

$$\partial_q([v_0,\dots, v_q])=\sum_{i=0}^q(-1)^i[v_0,\dots,\hat{v_i},\dots, v_q]$$
where $[v_0,\dots,\hat{v_i},\dots, v_q]$ denotes the oriented $(q-1)$-simplex obtained from removing the vertex $v_i$ in $[v_0,\dots, v_q]$.
\end{definition}

\begin{figure}[!htb]
\centering
\includegraphics[scale=1.2]{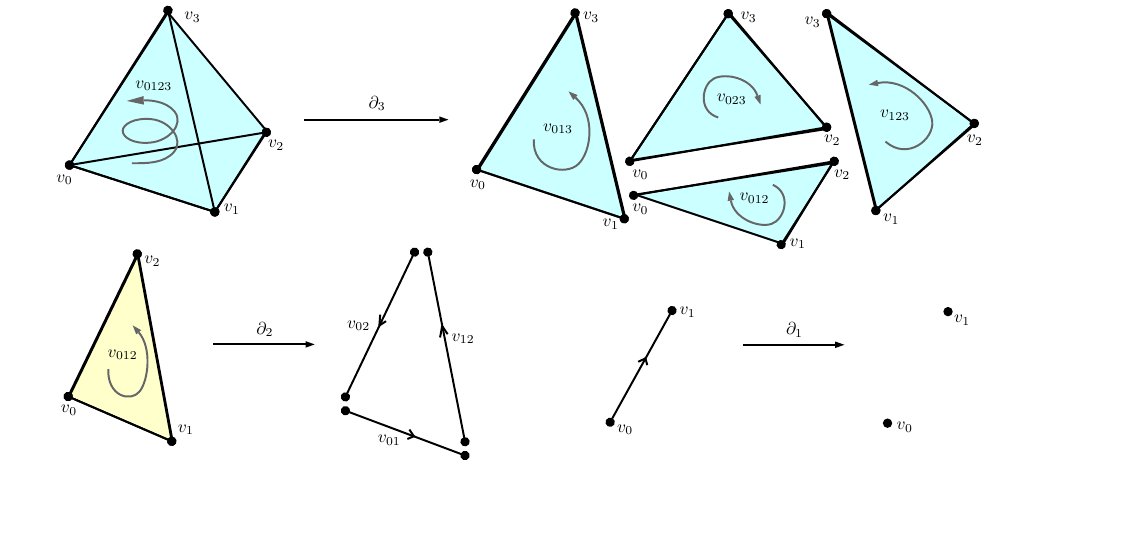}
\caption{Examples of $q$-boundary operators.}
\label{fig:Lq}
\end{figure}

\begin{remark} In Figure \ref{fig:Lq} we are denoting $[v_0,v_1,\dots, v_q]$ by $v_{01\cdots q}$.
\end{remark}

Let $C^q(K)=\Hom_k(C_q(K),k)$ the dual vector space of $C_q(K)$ (the field $k$ is allowed to be $\Z_p$, $\Q$, $\R$ or $\C$). Its elements, called \emph{cochains}, are completely determined by specifying its value on each simplex (since chains are linear combinations of simplices). Fixing an auxiliary inner product (with respect to which the basis of $C_q(K)$ can be chosen to be orthonormal), we can identify (via the associated polarity): $C_q(K)\simeq C^q(K)$. Then, we can define the \emph{coboundary operator}:
$$\delta_{q-1}\colon C^{q-1}(K) \to C^{q}(K)\,,$$
which is nothing but the \emph{adjoint operator}: 
$$\partial_q^*\colon C_{q-1}(K) \to C_{q}(K)$$
of the boundary map $\partial_q$. Given a linear form $\omega \in C^{q-1}(K)$ and an oriented simplex $\sigma=[v_0,\dots ,v_q]\in C_q(K)$, the map $\delta_{q-1}$ is defined as follows:
$$\delta_{q-1}(\omega)(\sigma)=\omega\Big(\sum_{i=0}^{q}(-1)^i[v_0,\dots ,\hat v_i,\dots  ,v_q]\Big),$$
which represents the evaluation of $\omega$ on all the faces of $\sigma$.

\begin{definition}
Given an oriented simplicial complex $K$, for $q\geq 0$ the $q$-combinatorial Laplacian is the linear operator $\Delta_q\colon C_q(K)\to C_q(K)$ defined by:
$$\Delta_q:=\partial_{q+1}\circ \partial_{q+1}^*+\partial_q^*\circ \partial_q\,.$$
The upper $q$-combinatorial Laplacian is defined as $\Delta_q^U:=\partial_{q+1}\circ \partial_{q+1}^*$ and the lower $q$-combinatorial Laplacian is defined as $\Delta_q^L:=\partial_q^*\circ \partial_q$.
\end{definition}

\begin{definition}\label{d:qUsimor}
Let $K$ be an oriented simplicial complex and $\sigma_i^{(q)}, \sigma_i^{(q)}\in C_q(K)$ two $q$-simplices which are upper adjacent with common upper $(q+1)$-simplex $\tau^{(q+1)}$. We say that $\sigma_i^{(q)}$ and $\sigma_j^{(q)}$ are similarly oriented with respect to $\tau^{(q+1)}$ if both have the same sign in $\partial_{q+1}(\tau^{(q+1)})$. If the signs are opposite, we say that they are upper dissimilarly oriented. 
\end{definition}
\begin{definition}\label{d:qLsimor}
Let $K$ be an oriented simplicial complex and $\sigma_i^{(q)}, \sigma_i^{(q)}\in C_q(K)$ two $q$-simplices which are lower adjacent with common lower $(q-1)$-simplex $\tau^{(q-1)}$. One says that $\sigma_i^{(q)}$ and $\sigma_j^{(q)}$ are similarly oriented with respect to $\tau^{(q-1)}$ if $\tau^{(q-1)}$ has the same sign in both $\partial_{q}(\sigma_i^{(q)})$ and $\partial_{q}(\sigma_j^{(q)})$. If the sign is different, they are lower dissimilarly oriented. 
\end{definition}

Let $B_q$ be the associated matrix to the boundary operator $\partial_q$ with respect to the orthonormal basis of elementary chains with some ordering for  $C_q(K)$ and $C_{q-1}(K)$. Then, the associated matrix to its adjoint operator $\partial_q^*=\delta_{q-1}$ with respect to the same ordered basis is the transpose matrix $B_q^t$. We call the \emph{$q$-Laplacian matrix of $K$} the associated matrix of $\Delta_q$:
\begin{equation}\label{e:Lq}
L_q=B_{q+1}B_{q+1}^t+B_q^tB_q\,,
\end{equation}
where we shall also denote $L_q^U=B_{q+1}B_{q+1}^t$ the upper $q$-Laplacian matrix and $L_q^L=B_q^tB_q$ the lower $q$-Laplacian matrix.

Then, the $q$-Laplacian matrices are given by (see for example \cite{G02}):
$$(L_q^U)_{ij}=\begin{cases}
\deg_U(\sigma_i^{(q)})  &\mbox{ if }\, i=j\,, \\
1  &\mbox{ if }\, i\neq j,\, \sigma_i^{(q)}\sim_U \sigma_j^{(q)} \, \mbox{ with similar orientation, }\\
-1  &\mbox{ if }\, i\neq j,\, \sigma_i^{(q)}\sim_U \sigma_j^{(q)} \, \mbox{ with dissimilar orientation, }\\
0  &\mbox{ otherwise. }
\end{cases}
$$
$$(L_q^L)_{ij}=\begin{cases}
\deg_L(\sigma_i^{(q)})=q+1  &\mbox{ if }\, i=j \,,\\
1  &\mbox{ if }\, i\neq j,\, \sigma_i^{(q)}\sim_L \sigma_j^{(q)} \, \mbox{ with similar orientation, }\\
-1  &\mbox{ if }\, i\neq j,\, \sigma_i^{(q)}\sim_L \sigma_j^{(q)} \, \mbox{ with dissimilar orientation, }\\
0  &\mbox{ otherwise. }
\end{cases}
$$

$$(L_q)_{ij}=\begin{cases}
\deg(\sigma_i^{(q)})  &\mbox{ if }\, i=j \,,\\
1  &\mbox{ if }\, i\neq j,\, \sigma_i^{(q)}\not\sim_U \sigma_j^{(q)} \,\mbox{ and }\, \sigma_i^{(q)}\sim_L \sigma_j^{(q)} \mbox{ with similar orientation. }\\
-1  &\mbox{ if }\, i\neq j,\, \sigma_i^{(q)}\not\sim_U \sigma_j^{(q)} \,\mbox{ and }\, \sigma_i^{(q)}\sim_L \sigma_j^{(q)} \mbox{ with dissimilar orientation.}\\
0  &\mbox{ if }\, i\neq j \mbox{ and either }\, \sigma_i^{(q)}\sim_U\sigma_j^{(q)} \,\mbox{ or }\, \sigma_i^{(q)}\not\sim_L \sigma_j^{(q)} \,.
\end{cases}
$$

\begin{remark}
The off diagonal entries of the $q$-adjacency matrix given in \cite{ER18} are the absolute value of the off diagonal entries of the $q$-Laplacian matrix $L_q$ just given following \cite{G02,MR12}. 
\end{remark}

\begin{remark}
Since a network is the $1$-skeleton of a simplicial complex and $\delta_0$ and $\partial_0^*$ are zero maps, then $L_0=\partial_1\circ \partial_1^*$ and $L_0=B_1B_1^t$ is given by:
$$(L_0)_{ij}=\begin{cases}
\deg (v_i)=\deg_U(v_i)  &\mbox{ if }\, i=j\,, \\-1  &\mbox{ if }\, v_i \, \mbox{ is upper adjacent to }\, v_j\,,\\0  &\mbox{ otherwise. }
\end{cases}
$$
Thus, $L_0=D-A$ is the traditional Laplacian matrix of a network.
\end{remark}

\section{Higher-order adjacency and simplicial degree}\label{s:qhdeg}\quad

In \cite{MDS15} the notion of upper degree for a simplex is further expanded. Their definition follows the idea that a vertex-to-triangle degree (that is the number of triangles connected to a vertex) can be computed by counting the number of triangles incident to each edge which is incident to the vertex, and then dividing by $2$ (since each triangle incident to the vertex has to be incident to a pair of edges connected to the vertex). Thus, they show that the ``vertex to triangle degree'' is given by: 
\begin{equation}\label{e:MDSv2deg}
\frac{1}{2}\sum_{j_1,j_2}|(B_1)_{i,j_1}||(B_2)_{j_1,j_2}|\,,
\end{equation}
where $B_1$ is the vertex-edge incidence matrix of equation  (\ref{e:Lq}).
 By an inductive procedure they propose the following formula for the ``vertex to $h$-simplex degree'':
\begin{equation}\label{e:MDSvhdeg}
\frac{1}{h!}\sum_{j_1,\dots,j_h}|(B_{1})_{i,j_1}||(B_{2})_{j_1,j_2}|\cdots |(B_{h})_{j_{h-1},j_h}|
\end{equation}
 
And, for the ``$q$-simplex to $(q+h)$-simplex degree'':
\begin{equation}\label{e:MDSqhdeg}
\frac{1}{h!}\sum_{j_1,\dots,j_h}|(B_{q+1})_{i,j_1}||(B_{q+2})_{j_1,j_2}|\cdots |(B_{q+h})_{j_{h-1},j_h}|
\end{equation}

They perform numerical analysis to study the associated degree distributions for the co-authorship network finding that, apart from the usual degree, there are not clear models for the ``$q$-simplex to $(q+h)$-simplex degree of a $q$-simplex''. Thus, they introduce an alternative extension of this higher-order notion of degree: the ``vertex to facets degree''. In social networks, a facet represents the number of different groups within which the social individual interacts, thus the ``vertex to facets degree of a vertex $v$'' is the number of distinct maximal collaborative groups which the vertex $v$ belongs to. That is, for each $h$ the ``vertex to facets degree of the vertex $v$'' defined in \cite{MDS15} is the number of $h$-simplices incident to the vertex $v$, and such that they are not incident to any other $(h+1)$-simplex incident to the vertex $v$: 
\begin{equation}\label{e:facetdeg}
\sum_{h\geq 1}\# \{\sigma^{(h)} \,|\, v\in \sigma^{(h)}\cap \sigma^{(h+1)} \, \mbox { and }\, \sigma^{(h)}\not\subset \sigma^{(h+1)}\}
\end{equation}

A generalization for ``$q$-simplex to facets degree'' is also given in \cite{MDS15}:
\begin{equation}\label{e:simplexfacetdeg}
\sum_{h\geq 1}\# \{\sigma^{(q+h)} \,|\, \sigma^{(q)}\subseteq \sigma^{(q+h)} \, \mbox { and }\, \sigma^{(q+h)}\not\subset \sigma^{(q+h+1)} \mbox{ if }\sigma^{(q)}\subseteq \sigma^{(q+h+1)}\}
\end{equation}

They affirm that, given a facet list, this degree can be ``computed with relatively straightforward searching and counting procedure'', but no explicit formula is given.

In this section we shall define a new notion of higher-order lower, upper and general adjacency for simplices and their associated degrees, which in particular allows us to redefine the ``$q$-simplex to $(q+h)$-simplex degree of a $q$-simplex'' of \cite{MDS15}. Then we will present properties and closed formulae for these degrees, and we shall also illustrate how to explicitly compute the ``$q$-simplex to facets degree'' by using our upper degree definition.

\subsection{General adjacencies for simplices of different dimensions.}\quad 
Let $\sigma_i^{(q)}$ be a $q$-simplex and $\sigma^{(q')}_j$ be a $q'$-simplex of a simplicial complex $K$. For simplicity we shall omit the subscripts $i$ and $j$, unless confusion can arise.
\begin{definition}\label{d:qhLoAdj}

We say that $\sigma^{(q)}$ and $\sigma^{(q')}$ are $p$-lower adjacent if there exists a $p$-simplex $\tau^{(p)}$ which is a common face of both $\sigma^{(q)}$ and $\sigma^{(q')}$:
$$\sigma^{(q)}\sim_{L_p}\sigma^{(q')}\,\iff\, \exists\,\, \tau^{(p)} \quad \colon\quad  \tau^{(p)}\subseteq \sigma^{(q)}\quad \&\quad  \tau^{(p)}\subseteq \sigma^{(q')}\,.$$

Note that if $\sigma^{(q)}\sim_{L_p}\sigma^{(q')}$, then $\sigma^{(q)}\sim_{L_{p'}}\sigma^{(q')}$ for all $0\leq p'\leq p$. Therefore, we say that $\sigma^{(q)}$ and $\sigma^{(q')}$ are strictly $p$-lower adjacent, referred as $p^*$-lower adjacent, if $\sigma^{(q)}\sim_{L_p}\sigma^{(q')}$ and $\sigma^{(q)}\not \sim_{L_{p+1}}\sigma^{(q')}$. We shall write $\sigma^{(q)}\sim_{L_{p^*}}\sigma^{(q')}$ for the strict lower adjacency.
\end{definition}

\begin{definition}\label{d:qhUpAdj}
We say that $\sigma^{(q)}$ and $\sigma^{(q')}$ are $p$-upper adjacent if there exists a $p$-simplex $\tau^{(p)}$ having both $\sigma^{(q)}$ and $\sigma^{(q')}$ as faces:
$$\sigma^{(q)}\sim_{U_p}\sigma^{(q')}\,\iff\, \exists\,\, \tau^{(p)} \quad \colon\quad  \sigma^{(q)}\subseteq \tau^{(p)} \quad \&\quad  \sigma^{(q')}\subseteq \tau^{(p)}\,.$$

We say that $\sigma^{(q)}$ and $\sigma^{(q')}$ are strictly $p$-upper adjacent, referred as $p^*$-upper adjacent and denoted as $\sigma^{(q)}\sim_{U_{p^*}}\sigma^{(q')}$, if $\sigma^{(q)}\sim_{U_p}\sigma^{(q')}$ and $\sigma^{(q)}\not \sim_{U_{p+1}}\sigma^{(q')}$.
\end{definition}

\begin{remark}
Let us point out some comments.
\begin{enumerate}
\item If $q=q'$ and $p=q-1$ (resp. $p=q+1$), then the notion $(q-1)$-lower (resp. $(q+1)$-upper) adjacency recovers the ordinary lower (resp. upper) adjacency for $q$-simplices of Definition \ref{d:qULAdj}. Thus, $(q+1)$-upper adjacency implies $(q-1)$-lower adjacency for $q$-simplices. However, in contrast to the ordinary case, the uniqueness of the common lower (resp. upper) simplex is no longer true.

\item If $h\geq 0$ and $\sigma^{(q)}\sim_{U_{q+h}}\sigma^{(q+h)}$, then $\sigma^{(q)}$ is a face of $\sigma^{(q+h)}$ and thus $\sigma^{(q)}\sim_{L_{q}}\sigma^{(q+h)}$. 
\item Although, in general is no longer true that $p$-upper adjacency implies $p'$-lower adjacency, one has that for $q\geq q'\geq h$ if $\sigma_i^{(q)}\sim_{U_{q+h}}\sigma_j^{(q')}$, then $\sigma_i^{(q)}\sim_{L_{q'-h}}\sigma_j^{(q')}$.
\end{enumerate}
\end{remark}

\begin{proposition}
Assume $\sigma_i^{(q)}\sim_{L_{p^*}}\sigma_j^{(q')}$ for some $p$ and put $p'=q+q'-p$. If $\sigma_i^{(q)}\not\sim_{U_{p'}}\sigma_j^{(q')}$, then $\sigma_i^{(q)}\not\sim_{U_{p'+h}}\sigma_j^{(q')}$ for all $h\geq 1$.
\end{proposition}
\begin{proof}
Assume $\sigma_i^{(q)}\sim_{U_{p'+h}}\sigma_j^{(q')}$ for some $h \geq 1$. Then, there exists a $(p'+h)$-simplex $\tau^{(p'+h)}$ such that $\sigma_i^{(q)}\cup\sigma_j^{(q')}\subseteq \tau^{(p'+h)}$. In particular, $\tau^{(p'+h)}$ has a $p'$-face containing both $\sigma_i^{(q)}$ and $\sigma_j^{(q')}$ as faces, that is, $\sigma_i^{(q+1)}\sim_{U_{p'}}\sigma_j^{(q')}$.
\end{proof}

\begin{remark} Notice that if $\sigma_i^{(q)}\sim_{L_{p^*}}\sigma_j^{(q')}$, then $\sigma_i^{(q)}$ and $\sigma_j^{(q')}$ share $p+1$-vertices. Thus, the smallest simplex which might contain both as faces (and therefore all of their vertices) has to have $q+1+q'+1-(p+1)=q+q'-p+1$ vertices, and thus it should be a $p'=q+q'-p$-simplex.
\end{remark}

This justifies the following definition.
\begin{definition}\label{d:qhAdj}
We say that $\sigma^{(q)}$ and $\sigma^{(q')}$ are $p$-adjacent if they are $p^*$-lower adjacent and not $p'$-upper adjacent, for $p'=q+q'-p$:
$$\sigma^{(q)}\sim_{A_p} \sigma^{(q')}\iff \sigma^{(q)}\sim_{L_{p^*}} \sigma^{(q')} \quad \& \quad \sigma^{(q)}\not\sim_{U_{p'}} \sigma^{(q')}.$$

In order to agree with graph theory, for $q=0$ we say that two vertices $v_i$ and $v_j$ are adjacent if $v_i\sim_{U_1} v_j$.

We say that $\sigma^{(q')}$ is maximal $p$-adjacent to $\sigma^{(q)}$ if:
$$\sigma^{(q')}\sim_{A_{p^*}} \sigma^{(q)}\iff \sigma^{(q')}\sim_{A_p} \sigma^{(q)} \quad \& \quad \sigma^{(q')}\not \subset \sigma^{(q'')} \quad \forall\quad \sigma^{(q'')}\,|\, \sigma^{(q'')}\sim_{A_p} \sigma^{(q)}\,.$$
\end{definition}

\begin{remark}
With the maximal $p$-adjacency we are saying that $\sigma^{(q')}$ and $\sigma^{(q)}$ are maximal collaborative simplicial communities in the sense that, even if some faces (sub-communities) of $\sigma^{(q')}$ might be $p$-adjacent to $\sigma^{(q)}$, they are not taken into account since the biggest one $p$-adjacent to $\sigma^{(q)}$ is $\sigma^{(q')}$.
\end{remark}

\subsection{Generalized lower degree for simplices.}\quad 

\begin{definition}\label{d:hpLdeg}
We define the $p$-lower degree of a $q$-simplex $\sigma^{(q)}$ as the number of $q'$-simplices which are $p$-lower adjacents to $\sigma^{(q)}$: 
$$\deg^{p}_L(\sigma^{(q)}):=\#\{\sigma^{(q')}\,\colon \, \sigma^{(q')}\sim_{L_p}\sigma^{(q)} \}\,.$$

The strictly $p$-lower degree of a $q$-simplex $\sigma^{(q)}$ is the number of $q'$-simplices which are $p^*$-lower adjacents to $\sigma^{(q)}$: 
$$\deg^{p^*}_L(\sigma^{(q)}):=\#\{\sigma^{(q')}\,\colon\, \sigma^{(q')}\sim_{L_p}\sigma^{(q)} \quad \& \quad \sigma^{(q')}\not\sim_{L_{p+1}}\sigma^{(q)}\}\,.$$

We define the $(h,p)$-lower degree of a $q$-simplex $\sigma^{(q)}$ as the number of $(q-h)$-simplices which are $p$-lower adjacents to $\sigma^{(q)}$: 
$$\deg^{h,p}_L(\sigma^{(q)}):=\#\{\sigma^{(q-h)}\,\colon\, \sigma^{(q-h)}\sim_{L_p}\sigma^{(q)}\}\,.$$
The strictly $(h,p)$-lower degree of a $q$-simplex $\sigma^{(q)}$ is the number of $(q-h)$-simplices which are $p^*$-lower adjacent to $\sigma^{(q)}$: 
$$\deg^{h,p^*}_L(\sigma^{(q)}):=\#\{\sigma^{(q-h)}\,\colon\, \sigma^{(q-h)}\sim_{L_p}\sigma^{(q)} \quad \& \quad \sigma^{(q-h)}\not\sim_{L_{p+1}}\sigma^{(q)}\}\,.$$
\end{definition}

\begin{figure}[!htb]
\centering
\includegraphics[scale=1.5]{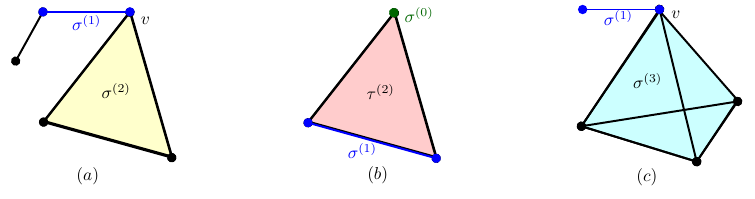}
\caption{Simplices and adjacency.}
\label{fig:Lower}
\end{figure}

\begin{example}
In Figure \ref{fig:Lower} (a) we have that the blue edge $\sigma^{(1)}$ is lower adjacent to the yellow triangle $\sigma^{(2)}$ (in the vertex $v$), but there does not exist $p'$ such that $\sigma^{(1)}$ and $\sigma^{(2)}$ were $p'$-upper adjacent. That is, $\sigma^{(1)}\sim_{L_0}\sigma^{(2)}$ and $\sigma^{(1)}\not\sim_{U_{p'}}\sigma^{(2)}$, and thus $\sigma^{(1)}\sim_{A_0} \sigma^{(2)}$. In this same picture, if we consider $q=1$, $h=0$ and $p=0$, then the $(0,0)$-lower degree of $\sigma^{(1)}$:
$$\deg_L^{0,0}(\sigma^{(1)})=\#\{\tau^{(1)}\,|\, \tau^{(1)}\sim_{L_0}\sigma^{(1)}\}$$
is the number of adjacent edges distinct to $\sigma^{(1)}$ in one of its vertices, so that it is $3$. If we choose $q=1$, $h=1$ and $p=0=q-h$, then the $(1,0)$-lower degree of $\sigma^{(1)}$:
$$\deg_L^{1,0}(\sigma^{(1)})=\#\{\tau^{(0)}\,|\, \tau^{(0)}\sim_{L_0}\sigma^{(1)}\}$$
is the number of vertices of $\sigma^{(1)}$, which is $2$.

Figure \ref{fig:Lower} (b) shows that the green vertex $\sigma^{(0)}$ is not lower adjacent to the blue edge $\sigma^{(1)}$ but they are $2$-upper adjacent, since there exists a triangle $\tau^{(2)}$ (in pink) containing both as faces. This gives an example that with our notion of adjacency, upper adjacency does not imply lower adjacency in general. 

In Figure \ref{fig:Lower} (c)  we have that $\sigma^{(1)}\sim_{L_0}\sigma^{(3)}$ (they intersect in the vertex $v$) but they are not upper adjacent, so that $\sigma^{(1)}\sim_{A_0} \sigma^{(3)}$. Setting $q=3$, $h=2$ and $p=0$ then $\deg_L^{2,0}(\sigma^{(3)})$ is the number of edges lower incident to the blue tetrahedron $\sigma^{(3)}$ in a vertex, ad thus it is $7$ (six coming from the edges of $\sigma^{(3)}$ plus the edge $\sigma^{(1)}$).
\end{example}

We have the following properties:
\begin{itemize}
\item If $h=1$ and $p=q-h=q-1$ then the $(1,q-1)$-lower degree of a $q$-simplex is the lower degree of the $q$-simplex of Definition \ref{d:qLUdeg}:
\begin{align*}
\deg^{1,q-1}_L(\sigma^{(q)})&:=\#\{\tau^{(q-1)}\,|\, \tau^{(q-1)}\sim_{L_{q-1}}\sigma^{(q)}\}=\\
&=\#\{(q-1)-\mbox{faces of }\sigma^{(q)}\}=q+1
\end{align*}

\item If $p=q-h$ we have that:
\begin{equation}\label{e:lowfaces}
\begin{aligned}
\deg^{h,q-h}_L(\sigma^{(q)})&:=\#\{\tau^{(q-h)}\,|\, \tau^{(q-h)}\sim_{L_{q-h}}\sigma^{(q)}\}=\\
&=\#\{(q-h)-\mbox{simplices of }\sigma^{(q)}\}=\binom{q+1}{q-h+1}
\end{aligned}
\end{equation}

\item From the very definition we have that: 
\begin{equation}\label{e:stlowdeg}
\deg^{h,p^*}_L(\sigma^{(q)})=\deg^{h,p}_L(\sigma^{(q)})-\deg^{h,p+1}_L(\sigma^{(q)})\,.
\end{equation}

\item $\deg^{p}_L(\sigma^{(q)})=\displaystyle\sum_{h=q-\dim K}^{q-p}\deg^{h,p}_L(\sigma^{(q)})$.
\item $\deg^{p^*}_L(\sigma^{(q)})=\displaystyle\sum_{h=q-\dim K}^{q-p}\deg^{h,p^*}_L(\sigma^{(q)})$.
\end{itemize}

\begin{figure}[!htb]
\centering
\includegraphics[scale=1.7]{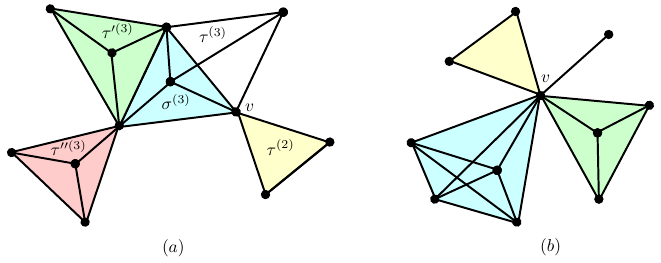}
\caption{Simplicial complexes. Computing lower and upper strict degrees.}
\label{fig:Sqhdeg}
\end{figure}

\begin{example}
Let us use Figure \ref{fig:Sqhdeg} (a) to perform a few computations for the lower degree of equation (\ref{e:stlowdeg}).
\begin{itemize}
\item Set $q=3$, $h=1$ and $p=1$. Let us compute the $(1,1)$-lower degree of the blue tetrahedron $\sigma^{(3)}$:
\begin{align*}
&\deg^{1,1}_L(\sigma^{(3)})=\#\{\mbox{incident triangles to }\sigma^{(3)} \mbox{ in an edge}\}=\\
&=\#\{\mbox{incident triangles to }\sigma^{(3)} \mbox{ in an edge and not in a triangle}\}+ \\
&+\#\{\mbox{triangles of }\sigma^{(3)}\}=\\ 
&=\deg^{1,1^*}_L(\sigma^{(3)})+\binom{q+1}{q-h+1}=\deg^{1,1^*}_L(\sigma^{(3)})+\deg^{1,2}_L(\sigma^{(3)})=\\
&=(2+3)+4=9\,.
\end{align*}
Where the $2$ comes from the two triangles of the green tetrahedron $\tau'^{(3)}$ incident to $\sigma^{(3)}$ in an edge and not in a triangle, the number $3$ comes from the three triangles of the white tetrahedron $\tau^{(3)}$ incident to $\sigma^{(3)}$ in an edge and not in a triangle, and finally the number  $4$ is the number of triangles in a tetrahedron.

\item Set $q=3$, $h=2$ and $p=0$. Let us compute the $(2,0)$-lower degree of the blue tetrahedron $\sigma^{(3)}$:
\begin{align*}
&\deg^{2,0}_L(\sigma^{(3)})=\#\{\mbox{incident edges to }\sigma^{(3)} \mbox{ in a vertex}\}=\\
&=\#\{\mbox{incident edges to }\sigma^{(3)} \mbox{ in a vertex and not in an edge}\}+ \\
&+\#\{\mbox{edges of }\sigma^{(3)}\}=\\ 
&=\deg^{2,0^*}_L(\sigma^{(3)})+\binom{q+1}{q-h+1}=\deg^{2,0^*}_L(\sigma^{(3)})+\deg^{2,1}_L(\sigma^{(3)})=\\
&=(3+4+3+2)+6=12+6=18\,.
\end{align*}
Where the first $3$ comes from the three edges of the white tetrahedron $\tau^{(3)}$ incident to $\sigma^{(3)}$ in a vertex and not in a edge, the number $4$ comes from the four edges of the green tetrahedron $\tau'^{(3)}$ incident to $\sigma^{(3)}$ in a vertex and not in a edge, the second $3$  comes from the three edges of the pink tetrahedron $\tau''^{(3)}$ incident to $\sigma^{(3)}$ in a vertex and not in a edge, the number $2$ comes from the two edges of the yellow triangle $\tau^{(2)}$ incident to $\sigma^{(3)}$ in a vertex, and $6$ is the number of edges in a tetrahedron.
\end{itemize}
\end{example}

\subsection{Generalized upper degree for simplices.}\quad 

\begin{definition}\label{d:hpUdeg}
We define the $p$-upper degree of a $q$-simplex $\sigma^{(q)}$ as the number of $q'$-simplices which are $p$-upper adjacent to $\sigma^{(q)}$: 
$$\deg^{p}_U(\sigma^{(q)}):=\#\{\sigma^{(q')}\,\colon\, \sigma^{(q)}\sim_{U_p}\sigma^{(q')}\}\,.$$

The strictly $p$-upper degree of a $q$-simplex $\sigma^{(q)}$ as the number of $q'$-simplices which are $p^*$-upper adjacent to $\sigma^{(q)}$: 
$$\deg^{p^*}_U(\sigma^{(q)}):=\#\{\sigma^{(q')}\,\colon\, \sigma^{(q)} \sim_{U_p}\sigma^{(q')}\quad \& \quad \sigma^{(q)}\not\sim_{U_{p+1}}\sigma^{(q')}\}\,.$$

We define the $(h,p)$-upper degree of a $q$-simplex $\sigma^{(q)}$ as the number of $(q+h)$-simplices which are $(h,p)$-upper adjacent to $\sigma^{(q)}$: 
$$\deg^{h,p}_U(\sigma^{(q)}):=\#\{\sigma^{(q+h)}\,\colon\, \sigma^{(q)}\sim_{U_p}\sigma^{(q+h)}\}\,.$$

The strictly $(h,p)$-upper degree of a $q$-simplex $\sigma^{(q)}$ as the number of $(q+h)$-simplices which are $p^*$-upper adjacent to $\sigma^{(q)}$: 

$$\deg^{h,p^*}_U(\sigma^{(q)}):=\#\{\sigma^{(q+h)}\,\colon\, \sigma^{(q)} \sim_{U_p}\sigma^{(q+h)}\quad \& \quad \sigma^{(q)}\not\sim_{U_{p+1}}\sigma^{(q+h)}\}\,.$$
\end{definition}

Notice that:

\begin{itemize}
\item For $q=0$, $h=1$ and $p=q+h=1$ then the $(1,1)$-upper degree of a vertex $v$ is the usual degree:
$$\deg^{1,1}_U(v)=\#\{\mbox {edges incident to }v\}=\deg (v)\,.$$
\item $\deg^{p}_U(\sigma^{(q)})=\displaystyle\sum_{h=-q}^{p-q}\deg^{h,p}_U(\sigma^{(q)})$.
\item $\deg^{p^*}_U(\sigma^{(q)})=\displaystyle\sum_{h=-q}^{p-q}\deg^{h,p^*}_U(\sigma^{(q)})$.
\end{itemize}

If $h>0$ and $p=q+h$ we recover the following known notions of degrees given in \cite{MDS15}:

\begin{itemize}
\item For $q=0$, $h=2$ and $p=q+h=2$ then the $(2,2)$-upper degree of a vertex $v$ is the ``vertex to triangle degree'' of equation (\ref{e:MDSv2deg}):
$$\deg^{2,2}_U(v)=\#\{\mbox {triangles incident to }v\}=\frac{1}{2}\sum_{j_1,j_2}|(B_1)_{i,j_1}||(B_2)_{j_1,j_2}|\,.$$
\item For $q=0$ and $p=q+h$ then the $(h,h)$-upper degree of a vertex $v$ is the ``vertex to $h$-simplex degree'' of equation (\ref{e:MDSvhdeg}):

\begin{align*}
\deg^{h,h}_U(v)&=\#\{(h)\mbox {-simplices incident to }v\}=\\
&=\frac{1}{h!}\sum_{j_1,\dots,j_h}|(B_{1})_{i,j_1}||(B_{2})_{j_1,j_2}|\cdots |(B_{h})_{j_{h-1},j_h}|\,.
\end{align*}
\item In general, for $p=q+h$, the $(h,q+h)$-upper degree of a $q$-simplex $\sigma^{(q)}$ is the ``$q$-simplex to $(q+h)$-simplex degree'' of equation (\ref{e:MDSqhdeg}):

\begin{align*}
\deg^{h,q+h}_U(\sigma^{(q)})&=\#\{(q+h)\mbox {-simplices incident to }\sigma^{(q)}\}=\\
&=\frac{1}{h!}\sum_{j_1,\dots,j_h}|(B_{q+1})_{i,j_1}||(B_{q+2})_{j_1,j_2}|\cdots |(B_{q+h})_{j_{h-1},j_h}|\,.
\end{align*}
In the following section, we will show a different way of computing these degrees by introducing a single new combinatorial Laplacian matrix.
\item The ``$q$-simplex to facets degree'' of equation (\ref{e:simplexfacetdeg}) can be given by the formula:
\begin{equation}\label{e:simplexfacetdeg2}
\sum_{h=1}^{\dim K-q}\deg^{h,(q+h)^*}_U(\sigma^{(q)})\,.
\end{equation}
\end{itemize}

Now, let us show how to compute the $(h,p^*)$-upper degree of a $q$-simplex, $\deg_U^{(h,p^*)}(\sigma^{(q)})$, and thus the facets degree. 

Let us consider the simplicial complex of Figure \ref{fig:Sqhdeg} (b), having a total of $11$ vertices ($0$-simplices), $20$ edges ($1$-simplices), $15$ triangles ($2$-simplices), $6$ tetrahedrons ($3$-simplices) and $1$ pentahedron ($4$-simplices). Its maximum dimension is therefore $4$. 
\begin{itemize}
\item Let us count the $(1,1^*)$-upper degree of the vertex $v$. We are in the case $q=0$, $h=1$ and $p=q+h=1$.  It is clear from the figure that the number of edges incident to $v$ which are not contained in a higher dimensional simplex is $1$. The count can be also performed in the following way: we start by counting the number of edges incident to $v$ (its usual degree, which is $10$), then many of these edges are contained in higher dimensional simplices, so we  star by subtracting the ones belonging to the same triangle still containing the vertex $v$, and we have to multiply this number by the number of triangles incident to $v$. That is, there are $10$ triangles incident to $v$ ($1$ alone in yellow, $3$ coming from the green tetrahedron and $6$ coming from the blue pentahedron) and there are always $2$ edges incident to a vertex in a triangle (thus we are subtracting $20$). At this step we have over subtracted edges incident to $v$: the ones that are also contained in a tetrahedron. Then we sum the number of incident tetrahedrons to $v$ (which are $1$ from the green tetrahedron and $3$ coming from the blue pentahedron) times the number of edges containing $v$ in a tetrahedron (which is always $3$). Again, we have over counted edges incident to $v$ which are contained in a pentahedron, thus we have to subtract the number of incident pentahedrons to $v$ (which is $1$, the blue one) times the number of edges containing $v$ in a pentahedron (which is $4$). We have now finished since the maximum dimension of the simplicial complex is $4$. Thus, we have:
{\small \begin{align*}
&\deg^{1,1^*}_U(v)=\#\{\mbox{incident edges to }v\}-\\
&-\#\{\mbox{incident triangles to }v\}\cdot \#\{\mbox{incident edges to }v\mbox{ in a triangle}\}+\\ 
&+\#\{\mbox{incident tetrahedrons to }v\}\cdot \#\{\mbox{incident edges to }v\mbox{ in a tetrahedron}\}-\\
&- \#\{\mbox{incident pentahedrons to }v\}\cdot \#\{\mbox{incident edges to }v\mbox{ in a pentahedron}\}=\\
&=10-(1+3+6)\cdot 2 + (1+4)\cdot 3-1\cdot 4=10-20+15-4=1\,.
\end{align*}}

\item Let us perform the computation for the $(2,2^*)$-upper degree of $v$, which is the number of incident triangles to $v$ which are not contained in a higher dimensional simplex (again from the figure we read that this number is $1$). We are in the case $q=0$, $h=2$ and $p=q+h=2$. 
{\small \begin{align*}
&\deg^{2,2^*}_U(v)=\#\{\mbox{incident triangles to }v\}-\\
&-\#\{\mbox{incident tetrahedrons to }v\}\cdot \#\{\mbox{incident triangles to }v\mbox{ in a tetrahedron}\}+\\
&+ \#\{\mbox{incident pentahedrons to }v\}\cdot \#\{\mbox{incident triangles to }v\mbox{ in a pentahedron}\}=\\
&=(1+3+6)- (1+4)\cdot 3-1\cdot 6=10-15+6=1\,.
\end{align*}}
\end{itemize}
A straightforward generalization of this strategy produces the following formula for the strict upper degree in terms of the upper degree (see definition \ref{d:hpUdeg}):

\begin{equation}\label{e:strictdegform}
\deg^{h,(q+h)^*}_U(\sigma^{(q)})=\sum_{i=0}^{\dim K-(q+h)} (-1)^i \deg^{h+i,q+h+i}_U(\sigma^{(q)})\cdot \binom{h+i}{h}\,.
\end{equation}

\begin{remark}
The combinatorial number $\binom{h+i}{h}=\binom{h+i}{i}$ comes from counting the number of $(q+h)$-simplices having $\sigma^{(q)}$ as a $q$-face in a single $(q+h+i)$-simplex. Since we have a total of $q+h+i+1$ points, and there are $q+1$ in $\sigma^{(q)}$, we have $(q+h+i+1)-(q+1)=h+i$ to form, joint with the $q+1$ points of $\sigma^{(q)}$, a $(q+h)$-simplex ($q+h+1$ points). Thus, we need to group the $h+i$ points in subsets of $i$ points.
 \end{remark}

Formula (\ref{e:strictdegform}) allows to compute the ``$q$-simplex to facets degree'' (a sum of the strict $(q,q+h)$-upper degrees) of equation (\ref{e:simplexfacetdeg2}) in terms of the $(h+i,q+h+i)$-upper degrees. We will show in the following section how to compute these generalized upper degrees as the diagonal entries of a single (multi parameter) combinatorial Laplacian matrix, instead of as a product of several entries of different matrices associated with distinct $q$-boundary operators, as stated in \cite{MDS15}.

\subsection{Generalized adjacency degree for simplices and simplicial degree.}\quad 

Let us finish this section by defining the generalized adjacency degrees associated with Definition \ref{d:qhAdj} and a general simplicial degree of a simplex.

\begin{definition}\label{d:Adjdeg}\quad 
\begin{enumerate}

\item We define the $p$-adjacency degree of a $q$-simplex $\sigma^{(q)}$ by:  
$$\deg^p_A(\sigma^{(q)}):=\#\{\sigma^{(q')} \,|\,\sigma^{(q)}\sim_{A_p} \sigma^{(q')}\}\,.$$

\item We define the maximal $p$-adjacency degree of a $q$-simplex $\sigma^{(q)}$ by:  
$$
\deg^{p*}_A(\sigma^{(q)}):=\#\{\sigma^{(q')} \,|\,\sigma^{(q')}\sim_{A_{p^*}} \sigma^{(q)}\}\,.
$$
\end{enumerate}
\end{definition}

\begin{remark}
Let us remark that with the $p$-adjacency degree we might be over counting certain simplices in the following sense: imaging we a triangle $\sigma^{(2)}$ to which another triangle $\sigma'^{(2)}$ is $0$-adjacent in a vertex $v$, then, since there are two edges ($1$-faces) of $\sigma'^{(2)}$ that are $0$-adjacent to $\sigma^{(2)}$, they are also being counted with the $p$-adjacency degree. That is, we are counting the community $\sigma'^{(2)}$ and two of its $1$-faces. This suggests that we should use in certain applications the maximal $p$-adjacency degree for a $q$-simplex, which only counts the maximal collaborative communities $p$-adjacent to a given simplex. 
\end{remark}

\begin{example}
In Figure \ref{fig:Sqhdeg} (a) we have that there are two triangles of $\tau^{'(3)}$ which are $1$-adjacent to $\sigma^{(3)}$, but which are not maximal $1$-adjacent to $\sigma^{(3)}$; The tetrahedron $\tau^{'(3)}$ is maximal $1$-adjacent to $\sigma^{(3)}$; the tetrahedron $\sigma^{(3)}$ is maximal $2$-adjacent to $\tau^{(3)}$.
\end{example}

If one would like to count all the collaborations of a simplex with different simplicial communities, both the ones collaborating with its faces and also the bigger simplicial communities on which the simplex is nested in, we can define a two parameter simplicial degree using both the adjacency degree and the upper degree as follows.

\begin{definition}\label{d:p1p2deg}
Given $p_1>q$ and $p_2<q$, we define the $(p_1,p_2^*)$-degree of a $q$-simplex $\sigma^{(q)}$ by:
$$\deg^{(p_1,p_2^*)}(\sigma^{(q)}):=\deg^{p_1}_U(\sigma^{(q)})+\deg^{p_2^*}_A(\sigma^{(q)})\,.$$
Similarly, for strict upper degree, we define the $(p_1^*,p_2^*)$-degree of a $q$-simplex $\sigma^{(q)}$ by:
$$\deg^{(p_1^*,p_2^*)}(\sigma^{(q)}):=\deg^{p_1^*}_U(\sigma^{(q)})+\deg^{p_2^*}_A(\sigma^{(q)})\,.$$
\end{definition}

Finally, let us propose a definition of maximal simplicial degree of a $q$-simplex, which counts all the maximal communities collaborating with the faces of the $q$-simplex (the ones that are maximal $p$-adjacent) and also the maximal communities to which the $q$-simplex belongs to (these last being strictly upper adjacent). That is to say, the maximal simplicial degree of a $q$-simplex counts the number of distinct facets which the $p$-faces of the $q$-simplex belongs to (for $p<q$) and also counts the distinct facets (different from the above ones) which the $q$-simplex belong to.

\begin{definition}\label{d:simpdeg}
We define the maximal simplicial degree of $\sigma^{(q)}$ by:
$$\deg^*(\sigma^{(q)})=\deg^*_A(\sigma^{(q)})+\deg^*_U(\sigma^{(q)})\,,$$ 
where: 
$$\deg^*_A(\sigma^{(q)}):=\sum_{p=0}^{q-1}\deg^{p^*}_A(\sigma^{(q)})\,;\quad \deg^*_U(\sigma^{(q)}):=\sum_{h=1}^{\dim K-q}\deg^{h,(q+h)^*}_U(\sigma^{(q)})\,.$$
We call  $\deg^*_U(\sigma^{(q)})$ the maximal upper simplicial degree of $\sigma^{(q)}$. 
\end{definition}

\begin{example}
Let $e$ be the edge ($1$-simplex) given by the intersection of the tetrahedra $\tau'^{(3)}$ and $\sigma^{(3)}$ in Figure \ref{fig:Sqhdeg} (a), $e=\tau'^{(3)} \cap \sigma^{(3)}$. By definition $e$ is contained in the facets $\tau'^{(3)}$ and $\sigma^{(3)}$, so its strict upper degree is $\deg^*_U(e)=2$. Moreover, the edge $e$ has two incident nodes (its $0$-faces) of which one is contained in the facet $\tau''^{(3)}$ and the other in the facet $\tau^{(3)}$, and thus its maximal adjacent degree is $\deg^*_A(e)=2$. Therefore, the maximal simplicial degree of $e$ is $\deg^*(e)=4$.
\end{example}

\section{The multi-combinatorial Laplacian}\label{s:qhLap}\quad 

We will define in this section generalized multi parameter boundary and coboundary operators in an oriented simplicial complex, and a new higher-order multi-combinatorial Laplacian will be introduced. They will give us a way to effectively compute all the higher-order degrees of the previous section.

\subsection{The generalized boundary operator.}

Let $\sigma^{(q)}$ be a $q$-simplex spanned by the set of points $\{v_0,\dots,v_q\}$. Given the $(q-h)$-face $\sigma^{(q-h)}$ of $\sigma^{(q)}$ spanned by the set of vertices $\{v_{0},\dots,\widehat v_{{j_1}},\dots,\widehat v_{{j_h}},\dots,v_{q}\}$ let us  denote by $\epsilon_{j_1\cdots j_h}$ the permutation 
$$\begin{pmatrix}
0&\cdots& h-1&h&\cdots&q \\ 
{j_1}&\cdots&  {j_h}&0&\cdots & {q}
\end{pmatrix}\,.$$  As oriented $q$-simplex, $\sigma^{(q)}$ is represented by $[v_{\eta(0)},\dots,v_{\eta(q)}]$, for some permutation $\eta$ in the set of its vertices. 

\begin{definition}\label{d:qhBoundOp}
We define the {$(q,h)$-boundary operator} $$\partial_{q,h}\colon \calC_q(K)\to \calC_{q-h}(K)$$ as the homomorphism given as the linear extension of:

$$\partial_{q,h}([v_{\eta(0)},\dots, v_{\eta(q)}])=\sum_{j_{1},\dots,j_{h}} \sign(\eta)\sign(\epsilon_{j_1\cdots j_h})[v_0,\dots,\widehat{v_{j_1}},\dots, \widehat{v_{j_{h}}},\dots ,v_q]$$
where $[v_0,\dots,\widehat{v_{j_1}},\dots, \widehat{v_{j_{h}}},\dots ,v_q]$ denotes the oriented $q$-simplex obtained from removing the vertices $v_{j_1},\dots, v_{j_{h}}$ in $[v_0,\dots, v_q]$. 
\end{definition}

Note that $[v_{\eta(0)},\dots,v_{\eta(q)}]=[v_{\eta'(0)},\dots,v_{\eta'(q)}]$ if and only if $\sign (\eta)=\sign (\eta')$, so that $[v_{\eta(0)},\dots,v_{\eta(q)}]=\sign(\eta)\,[v_0,\dots,v_q]$. Then, this operator is well defined and $\partial_{q,h}(-\sigma^{(q)})=-\partial_{q,h}(\sigma^{(q)})$. Moreover, for $h=1$ the operator $\partial_{q,h}$ is the ordinary $q$-boundary operator $\partial_q$.

\begin{figure}[!htb]
\centering
\includegraphics[scale=1.3]{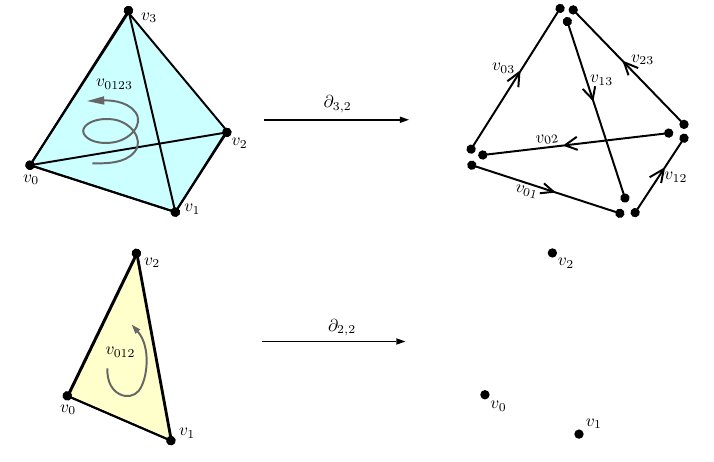}
\caption{Examples of $(q,h)$-boundary operators.}
\label{fig:Lqh}
\end{figure}

Given $\tau^{(p)}$ a $p$-simplex and $\sigma^{(q)}$ a $q$-face in $K$, with $q<p$, we denote by $\sign\big(\tau^{(p)},\sigma^{(q)}\big)$ the coefficient of $\sigma^{(q)}$ in the sum $\partial_{p,p-q}(\tau^{(p)})$.

\begin{definition}
Let $\sigma_i^{(q)}$ and $\sigma_j^{(q')}$ be two simplices which are $p$-upper adjacent. Let $\tau^{(p)}$ be a common upper $p$-simplex. We say that $\sigma_i^{(q)}$ and $\sigma_j^{(q')}$ are {upper similarly oriented with respect to $\tau^{(p)}$} if $\sign\big(\tau^{(p)},\sigma^{(q)}\big)=\sign\big(\tau^{(p)},\sigma^{(q')}\big)$.

We shall denote it by $\sigma_i^{(q)}\sim_{U_{\tau^{(p)}}^+}\sigma_j^{(q')}$. If the signs are different, we say that they are { dissimilarly oriented} with respect to $\tau^{(p)}$. We shall denote it by $\sigma_i^{(q)}\sim_{U_{\tau^{(p)}}^-}\sigma_j^{(q')}$.
\end{definition}
\begin{remark}
The equality or inequality of $\sign\big(\tau^{(p)},\sigma^{(q)}\big)$ and $\sign\big(\tau^{(p)},\sigma^{(q')}\big)$ does not depend on the orientation of $\tau^{(p)}$ but only on the orientations of $\sigma^{(q)}$ and $\sigma^{(q')}$.
\end{remark}
\begin{remark}
For $h=1$ this definition recovers Definition \ref{d:qUsimor}.
\end{remark}

Let $\sigma_i^{(q)},\sigma_j^{(q')}$  and $\tau^{(p)}$ oriented simplices.
\begin{definition}\label{d:psigU}
 We define the {upper sign} of $\sigma_i^{(q)}$ and $\sigma_j^{(q')}$ with respect to $\tau^{(p)}$ as the following function:
$$
\sig_U(\sigma_i^{(q)},\sigma_j^{(q')};\tau^{(p)}):=\begin{cases}
0 & \mbox{ if } \sigma_i^{(q)}\cup\sigma_j^{(q')}\nsubseteq\tau^{(p)}\\
1 & \mbox{ if } \sigma_i^{(q)}\sim_{U_{\tau^{(p)}}^+}\sigma_j^{(q')} 
\\
-1 & \mbox{ if } \sigma_i^{(q)}\sim_{U_{\tau^{(p)}}^-}\sigma_j^{(q')}
\end{cases}
$$
\end{definition}
Note that if $\tau^{(p)}$ is a common upper $p$-simplex to $\sigma_i^{(q)}$ and $\sigma_j^{(q')}$, then $$\sig_U(\sigma_i^{(q)},\sigma_j^{(q')};\tau^{(p)})=\sign\big(\tau^{(p)},\sigma^{(q)}\big)\cdot\sign\big(\tau^{(p)},\sigma^{(q')}\big)\,,$$ which does not depend on the orientation of $\tau^{(p)}$. This justifies the following definition.
\begin{definition}
Let $\sigma_i^{(q)}$ and $\sigma_j^{(q')}$ two simplices. 

We define the $p$-{upper oriented degree} of $\sigma_i^{(q)}$ and $\sigma_j^{(q')}$ as the following sum:
$$\odeg^p_U(\sigma_i^{(q)},\sigma_j^{(q')}):=\frac{1}{2}\sum_{\tau^{(p)}\in \widetilde{S}_p(K)}\sig_U(\sigma_i^{(q)},\sigma_j^{(q')};\tau^{(p)})\,.$$
where by $\widetilde{S}_p(K)$ we denote the set of oriented $p$-simplices of the simplicial complex $K$.
\end{definition}

\begin{remark}
Note that we are dividing by $2$ since, as mentioned above, we have $\sig_U(\sigma_i^{(q)},\sigma_j^{(q')};\tau^{(p)})=\sig_U(\sigma_i^{(q)},\sigma_j^{(q')};-\tau^{(p)})$.
\end{remark}

Let us point out that if $p=q+1$ and $\sigma_i^{(q)}\sim_{U_{q+1}}\sigma_j^{(q)}$, then there exists a unique common $(q+1)$-simplex $\tau^{(q+1)}$, and thus:
\begin{align*}
\odeg^{q+1}_U&(\sigma_i^{(q)},\sigma_j^{(q)})=\sig_U(\sigma_i^{(q)},\sigma_j^{(q)};\tau^{(q+1)})=\\
&=\begin{cases}
1 & \mbox{ if } \sigma_i^{(q)}\sim_{U_{q+1}}\sigma_j^{(q)} \mbox{ and similarly oriented w.r.t. }\tau^{(q+1)}\,.\\
-1 & \mbox{ if } \sigma_i^{(q)}\sim_{U_{q+1}}\sigma_j^{(q)} \mbox{ and dissimilarly oriented w.r.t. }\tau^{(q+1)}\,,
\end{cases}
\end{align*}
which therefore recovers the notion of similarly and dissimilarly oriented of \cite{G02}.

Let us now give the analogous definition for the lower adjacency. 

\begin{definition}
Let $\sigma_i^{(q)}$ and $\sigma_j^{(q')}$ two simplices which are $p$-lower adjacent and  $\tau^{(p)}$ be a common lower $p$-face. We say that $\sigma_i^{(q)}$ and $\sigma_j^{(q')}$ are { lower similarly oriented with respect to $\tau^{(p)}$} if $\sign\big(\sigma^{(q)},\tau^{(p)}\big)=\sign\big(\sigma^{(q')},\tau^{(p)}\big)$. If the signs are different, we say that they are {dissimilarly oriented} with respect to $\tau^{(p)}$. As before, we shall denote it by $\sigma_i^{(q)}\sim_{L_{\tau^{(p)}}^+}\sigma_j^{(q')}$ and $\sigma_i^{(q)}\sim_{L_{\tau^{(p)}}^-}\sigma_j^{(q')}$, respectively.
\end{definition}

\begin{remark}
The equality or inequality of $\sign\big(\sigma^{(q)},\tau^{(p)}\big)$ and $\sign\big(\sigma^{(q')},\tau^{(p)}\big)$ does not depend on the orientation of $\tau^{(p)}$ but only on the orientations of $\sigma^{(q)}$ and $\sigma^{(q')}$.
\end{remark}

\begin{remark}
For $p=q-1$ this definition recovers Definition \ref{d:qLsimor}.
\end{remark}

\begin{definition}\label{d:psigL}
Let $\sigma_i^{(q)}$ and $\sigma_j^{(q')}$ two simplices. We define the {lower sign} of $\sigma_i^{(q)}$ and $\sigma_j^{(q')}$ with respect to a $p$-simplex $\tau^{(p)}$ as the following function:
$$
\sig_L(\sigma_i^{(q)},\sigma_j^{(q')};\tau^{(p)}):=\begin{cases}
0 & \mbox{ if } \tau^{(p)}\nsubseteq\sigma_i^{(q)}\cap\sigma_j^{(q')}\\
1 & \mbox{ if } \sigma_i^{(q)}\sim_{L_{\tau^{(p)}}^+}\sigma_j^{(q')} 
\\
-1 & \mbox{ if } \sigma_i^{(q)}\sim_{L_{\tau^{(p)}}^-}\sigma_j^{(q')} 
\end{cases}
$$

\end{definition}
As we pointed out for the upper sign, notice that if $\tau^{(p)}$ is a common lower $p$-face of $\sigma_i^{(q)}$ and $\sigma_j^{(q')}$, then $$\sig_L(\sigma_i^{(q)},\sigma_j^{(q')};\tau^{(p)})=\sign\big(\sigma^{(q)},\tau^{(p)}\big)\cdot\sign\big(\sigma^{(q')},\tau^{(p)}\big)\,,$$ which does not depend on the orientation of $\tau^{(p)}$. This justifies the following definition.
\begin{definition}
Let $\sigma_i^{(q)}$ and $\sigma_j^{(q')}$ two simplices. We define the $p$-{lower oriented degree} of $\sigma_i^{(q)}$ and $\sigma_j^{(q')}$ as the following sum:
$$\odeg^p_L(\sigma_i^{(q)},\sigma_j^{(q')}):=\frac{1}{2}\sum_{\tau^{(p)}\in \widetilde{S}_p(K)}\sig_L(\sigma_i^{(q)},\sigma_j^{(q')};\tau^{(p)})\,.$$
where by $\widetilde{S}_p(K)$ we denote the set of oriented $p$-simplices of the simplicial complex $K$.
\end{definition}

\begin{remark}
Note that we are dividing by $2$ since, as mentioned above, we have $\sig_L(\sigma_i^{(q)},\sigma_j^{(q')};\tau^{(p)})=\sig_L(\sigma_i^{(q)},\sigma_j^{(q')};-\tau^{(p)})$.
\end{remark}

\begin{figure}[!htb]
\centering
\includegraphics[scale=1.]{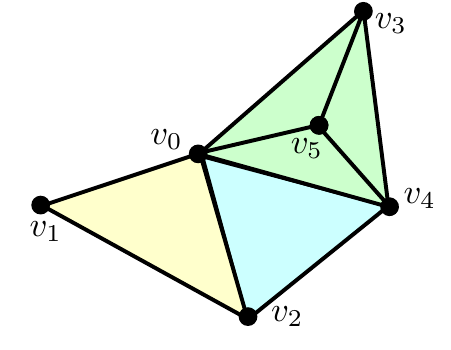}
\caption{Simplicial complexes. Computing oriented degrees.}
\label{fig:Odeg}
\end{figure}

\begin{example}\label{ex:odeg}
Let $K$ be the simplicial complex given by Figure \ref{fig:Odeg}. Let $$C_0(K)=\{v_0,v_1,\dots,v_5\}$$ be the oriented set of vertices and $$C_2(K)=\langle v_{012},v_{024},v_{034},v_{035},v_{045},v_{345}\rangle$$ the free abelian group generated by triangles (choosing one of its orientations) where we are writing $v_{ijk}$ for the oriented triangle $[v_i,v_j,v_k]$). Let 
\begin{align*}
\partial_{2,2} \colon C_2(K)&\to C_0(K)\\
v_{ijk}&\mapsto  \sign(\bar\epsilon_{j,k})v_{i}+ \sign(\bar\epsilon_{i,k})v_{j}+ \sign(\bar\epsilon_{i,j})v_{k}
\end{align*}
be the $(2,2)$-boundary operator, with
$$\bar\epsilon_{j,k}=
\begin{pmatrix}
i&j&k \\ 
j&k&i
\end{pmatrix},\, \bar\epsilon_{i,k}=
\begin{pmatrix}
i&j&k \\ 
i&k&j
\end{pmatrix}\text{ y }\,\bar\epsilon_{i,j}=
\begin{pmatrix}
i&j&k \\ 
i&j&k
\end{pmatrix}$$
Thus: 
$$\partial_{2,2}(v_{ijk})=v_i-v_j+v_k$$
We have that:

\begin{align*}
\odeg_U^2(v_0,v_2)&:=\sum_{\tau^{(2)}\in C_2(K)}\sig_U(v_0,v_2;\tau^{(2)})=\\
&\,\,=\sig_U(v_0,v_2;v_{012})+\sig_U(v_0,v_2;v_{024})=1-1=0
\end{align*}
\begin{align*}
\odeg_L^0(v_{012},v_{024})&:=\sum_{\tau^{(0)}\in C_0(K)}\sig_L(v_{012},v_{024};\tau^{(0)})=\\
&\,\,=\sig_L(v_{012},v_{024};v_0)+\sig_L(v_{012},v_{024};v_2)=1-1=0
\end{align*}

\begin{align*}
\odeg_U^2(v_0,v_3)&:=\sum_{\tau^{(2)}\in C_2(K)}\sig_U(v_0,v_3;\tau^{(2)})=\\
&\,\,=\sig_U(v_0,v_3;v_{034})+\sig_U(v_0,v_3;v_{035})=-1-1=-2\,.
\end{align*}
\end{example}

Assume that $\tau^{(q)}$ is a $q$-simplex and $\sigma^{(p)}$ is a $p$-face of $\tau^{(q)}$. By the above definitions we have:
\begin{enumerate}\itemsep.25cm
\item $\odeg_U^q(\sigma^{(p)},\tau^{(q)})=\sig_U(\sigma^{(p)},\tau^{(q)};\tau^{(q)})$.
\item $\odeg_L^p(\sigma^{(p)},\tau^{(q)})=\sig_L(\sigma^{(p)},\tau^{(q)};\sigma^{(p)})$.
\item $\sig_U(\sigma^{(p)},\tau^{(q)};\tau^{(q)})=\sign\big(\tau^{(q)},\sigma^{(p)}\big)=\sig_L(\sigma^{(p)},\tau^{(q)};\sigma^{(p)})$.
\end{enumerate}

Then one obtains the following result. 
\begin{proposition}\label{p:bounop}
\begin{equation}\label{eq:boundary}
\partial_{q,h}(\tau^{(q)})=\displaystyle\sum_{\sigma^{(q-h)}\in S_{q-h}(K)}\odeg_L^{q-h}(\tau^{(q)},\sigma^{(q-h)})\sigma^{(q-h)}.
\end{equation}
\end{proposition}

Let us now show that there exists a coboundary operator in the oriented simplicial complex $K$, and that it can be written down in terms of the oriented degree. 

\begin{proposition}\label{p:cobounop}
Given an oriented simplicial complex $K$, let $\partial_{_{q,h}}$ be the $(q,h)$-boundary operator. There exists a unique homomorphism: $$\partial_{q,h}^*\colon C_{q-h}(K)\to C_q(K)$$ 
defined as: 
\begin{equation}\label{eq:adjoint}
\partial_{q,h}^*(\sigma^{(q-h)})=\displaystyle\sum_{\tau^{(q)}\in S_q(K)}\odeg_U^{q}(\sigma^{(q-h)},\tau^{(q)})\tau^{(q)}
\end{equation}
and such that $\partial_{_{q,h}}$ and $\partial_{q,h}^*$ are adjoint operators.
\end{proposition}
\begin{proof}

Consider $\{\tau_1,\dots,\tau_n\}$ and $\{\sigma_1,\dots,\sigma_m\}$ basis of $C_q(K)$ and $C_{q-h}(K)$ respectively. For all $\sigma_i\subseteq \tau_j$ we have 
$$\langle\sig_U(\sigma_i,\tau_j;\tau_j)\tau_j,\tau_k\rangle=\langle\sig_L(\sigma_i,\tau_j;\sigma_i)\tau_j,\tau_k\rangle=\langle\tau_j,\sig_L(\sigma_i,\tau_j;\sigma_i)\tau_k\rangle$$
where $\langle\,,\,\rangle$ denote the standard inner product 
$\Big\langle \sum_i\alpha_i\tau_i,\sum_j\beta_j\tau_j\Big\rangle=
\sum_{k}\alpha_k\beta_k\,.$

Therefore, 
$$
\begin{aligned}
\Big\langle\sum_{\tau_j}\odeg_U^q(\sigma_i,\tau_j)\tau_j,\tau_k\Big\rangle&=\sig_U(\sigma_i,\tau_k;\tau_k)=\sig_L(\sigma_i,\tau_k;\sigma_i)=\\
&=\Big\langle\sigma_i,\sum_{\sigma_j}\odeg_L^{q-h}(\sigma_j,\tau_k)\sigma_j\Big\rangle=\langle\sigma_i,\partial_{q,h}(\tau_k)\rangle
\end{aligned}$$
and the result follows.
\end{proof}

\subsection{The multi-combinatorial Laplacian.}
Once we have the boundary and coboundary operators in an oriented simplicial complex $K$, we can define a multi-combinatorial Laplacian and show how it computes some higher-order degrees of simplices.  

\begin{definition}\label{d:multiLap}
Let $q,h,h'$ non-negative integers. We define the {$(q,h,h')$-Laplacian operator} $$\Delta_{q,h,h'}\colon C_q(K)\to C_q(K)$$ as the following operator:
$$\Delta_{q,h,h'}:=\partial_{q+h,h}\circ \partial^*_{q+h,h}+\partial^*_{q,h'}\circ \partial_{q,h'}\,.$$

$\Delta^U_{q,h}=\partial_{q+h,h}\circ \partial^*_{q+h,h}\colon C_{q}(K)\to C_q(K)$ is named the upper $(q,h)$-Laplacian operator and $\Delta^L_{q,h'}=\partial^*_{q,h'}\circ \partial_{q,h'}\colon C_q(K)\to C_{q}(K)$ is called the $(q,h')$-Laplacian operator.
\end{definition}

Let us fix basis of ordered simplices for $C_{q+h}(K),\, C_{q}(K)$ and $C_{q-h'}(K)$ and denote by $B_{q+h,h}$ and $B_{q,h'}$ the corresponding matrix representation of $\partial_{q+h,h}\colon C_{q+h}(K)\to C_{q}(K)$ and $\partial_{q,h'}\colon C_{q}(K)\to C_{q-h'}(K)$, respectively. Then the associated matrix of the $(q,h,h')$-Laplacian operator is 
$$L_{q,h,h'}=B_{q+h,h}\,B^t_{q+h,h}+B^t_{q,h'}\,B_{q,h'}\,.$$
We shall call it the $(q,h,h')$-Laplacian matrix and as before we use the notation  $L_{q,h}^U=B_{q+h,h}\,B^t_{q+h,h}$ and $L_{q,h'}^L=B^t_{q,h'}\,B_{q,h'}$.
\begin{theorem}\label{thm:Laplacian}
Let $K$ be an oriented simplicial complex and fix oriented basis on the $q$-chains $C_q(K)$ of $K$. With respect to this basis, the $(i,j)$-th entry of the associated matrices of the multi-combinatorial Laplacian operators is given by:
\begin{align*}
\big(L^U_{q,h}\big)_{i,j}&=\begin{cases}
\deg_U^{h,q+h}(\sigma_i^{(q)}) & \mbox{ if } i=j \\
\odeg_U^{q+h}(\sigma_i^{(q)},\sigma_j^{(q)}) & \mbox{ if } i\neq j
\end{cases}\\
&\\
\big(L^L_{q,h'}\big)_{i,j}&=\begin{cases}
\deg_L^{h',q-h'}(\sigma_i^{(q)})=\binom{q+1}{q-h'+1} & \mbox{ if } i=j \\
\odeg_L^{q-h'}(\sigma_i^{(q)},\sigma_j^{(q)}) & \mbox{ if } i\neq j
\end{cases}\\
&\\
\big(L_{q,h,h'}\big)_{i,j}&=\begin{cases}
\deg_U^{h,q+h}(\sigma_i^{(q)})+\binom{q+1}{q-h'+1} & \mbox{ if } i=j \\
\odeg_U^{q+h}(\sigma_i^{(q)},\sigma_j^{(q)})+\odeg_L^{q-h'}(\sigma_i^{(q)},\sigma_j^{(q)}) & \mbox{ if } i\neq j
\end{cases}
\end{align*}
\end{theorem}
\begin{proof}
Fix $\{\tau_1^{(q+h)},\dots,\tau_m^{(q+h)}\},\,\{\sigma_1^{(q)},\dots,\sigma_n^{(q)}\}$ and $\{\gamma_1^{(q-h')},\dots,\gamma_r^{(q-h')}\}$ basis of $C_{q+h}(K)$, $C_{q}(K)$ and $C_{q-h'}(K)$ respectively. By using the above notations one has that the $(i,j)$-th entry of $B_{q,h'}$ is 
$$b_{ij}^{(q,h')}=\langle\partial_{q,h'}(\sigma_j^{(q)}),\gamma_i^{(q-h')}\rangle=\sig_L\big(\sigma_j^{(q)},\gamma_i^{(q-h')};\gamma_i^{(q-h')}\big)\,.$$
Then 
$$
\begin{aligned}
\big(L_{q,h'}^L\big)_{i,j}&=\sum_{k=1}^r b_{ki}^{(q,h')}b_{kj}^{(q,h')}\\&=\sum_{k=1}^r\sig_L\big(\sigma_i^{(q)},\gamma_k^{(q-h')};\gamma_k^{(q-h')}\big)\sig_L\big(\sigma_j^{(q)},\gamma_k^{(q-h')};\gamma_k^{(q-h')}\big)\\
&= 
\begin{cases}
\displaystyle\sum_{k=1}^r\sig_L\big(\sigma_i^{(q)},\sigma_j^{(q)};\gamma_k^{(q-h')}\big)& \text{ for } i\neq j\\
\\
\deg_L^{h',q-h'}(\sigma_i^{(q)}) & \text{ if } i=j 
\end{cases}
\end{aligned}
$$
The explicit description of the adjoint operator of $\partial_{q+h,h}$ previously given in Equation \ref{eq:adjoint} shows that the $(i,j)$-th entry of $B^t_{q+h,h}$ is 
$${b^*}_{ij}^{(q+h,h)}=\langle\partial^*_{q+h,h}(\sigma_j^{(q)}),\tau_i^{(q+h)}\rangle=\sig_U\big(\sigma_j^{(q)},\tau_i^{(q+h)};\tau_i^{(q+h)}\big)$$
so, we get
$
\big(L^U_{q,h}\big)_{i,j}=
\begin{cases}
\displaystyle\sum_{l=1}^m\sig_U\big(\sigma_i^{(q)},\sigma_j^{(q)};\tau_l^{(q+h)}\big)&\text{ for } i\neq j\\
\\
\deg_U^{h,q+h}(\sigma_i^{(q)})&\text{ if }i=j
\end{cases}
$
\end{proof}

Notice that for $h=h'=1$ we recover the $q$-combinatorial Laplacian (see \cite{G02,MR12}).

\begin{remark}
As opposite with the $q$-combinatorial Laplacian matrix and the graph Laplacian matrix, there might be $0$ entries in the multi-combinatorial Laplacian matrix coming not only from non-adjacent simplices, but also from simplices which being, for example, lower adjacent, the orientation of a common face is opposite to one another, and thus it cancels the corresponding oriented degree. See for instance the second computation of Example \ref{ex:odeg}.
\end{remark}

\begin{example}
With the notations of example \ref{ex:odeg}, let us compute an upper and a lower multi-combinatorial Laplacian of the simplicial complex $K$ of Figure \ref{fig:Odeg}. Recall that we have the basis $\{v_{012},v_{024},v_{034},v_{035},v_{045},v_{345}\}$ of $C_2(K)$ and $\{v_0,\dots, v_5\}$ of $C_0(K)$.
\begin{itemize}
\item Let us set $q=0$ and $h=2$. The associated matrix of the $(q+h,h)=(2,2)$-boundary operator $\partial_{2,2}$ is:
$$B_{2,2}=\begin{pmatrix}
1&1&1&1&1&0\\
-1&0&0&0&0&0\\
1&-1&0&0&0&0\\
0&0&-1&-1&0&1\\
0&1&1&0&-1&-1\\
0&0&0&1&1&1
\end{pmatrix}$$
and, since the associated matrix (with respect to the corresponding dual basis) of its adjoint operator $\partial^*_{2,2}$ is its transpose, the matrix of the $(q,h)=(0,2)$-upper Laplacian operator $L^U_{0,2}=\partial_{2,2}\circ \partial^*_{2,2}\colon C_0(K)\to C_0(K)$  is:
$$B_{2,2}\cdot B^t_{2,2}=\begin{pmatrix}
5&-1&0&-2&1&2\\
-1&1&-1&0&0&0\\
0&-1&2&0&-1&0\\
-2&0&0&3&-2&0\\
1&0&-1&-2&4&-2\\
2&0&0&0&-2&3
\end{pmatrix}$$
whose diagonal entries are the upper degrees of the vertices $\deg_U^{h,q+h}(v_i)=\deg_U^{2,2}(v_i)$ (for $i=0,1,\dots ,5$), and the off diagonal entries are the upper oriented degrees $\odeg_U^{2}(v_i,v_j)$ for $i\neq j$ (see example \ref{ex:odeg}).

\item Similarly, let us set $q=2$ and $h=2$, then the matrix of the $(q,h)=(2,2)$-boundary operator $\partial_{2,2}$ is again $B_{2,2}$, and the matrix of the $(q,h)=(2,2)$-lower Laplacian operator $L^L_{2,2}=\partial^*_{2,2}\circ \partial_{2,2}\colon C_2(K)\to C_2(K)$  is:
$$B^t_{2,2}\cdot B_{2,2}=\begin{pmatrix}
3&0&1&1&1&0\\
0&3&2&1&0&-1\\
1&2&3&2&0&-2\\
1&1&2&3&2&0\\
1&0&0&2&3&2\\
0&-1&-2&0&2&3
\end{pmatrix}$$
whose diagonal entries are the lower degrees of the vertices $\deg_L^{h,q-h}(v_{ijk})=\deg_L^{2,0}(v_{ijk})$, and the off diagonal entries are given by the lower oriented degrees $\odeg_L^{0}(v_{ijk},v_{i'j'k'})$ (see example \ref{ex:odeg}).

\end{itemize}
\end{example}

\subsection{The boundary and coboundary operators compute the higher-order degrees.}

Notice that the multi-combinatorial Laplacian of Theorem \ref{thm:Laplacian} does not compute all the higher-order degrees of simplices, for instance the general $p$-lower degree of a simplex $\sigma^{(q)}$ is not computed, only its lower degree $\deg_L^{h',q-h'}(\sigma^{(q)})$ is contained in the multi-combinatorial Laplacian, and we already knew that it were equal to $\binom{q+1}{q-h'+1}$. Let us finish this section by giving an explicit description of all the higher-order degrees in terms of the generalized boundary and coboundary operators. The key point to perform these computations is to use Propositions \ref{p:bounop} and \ref{p:cobounop}.

We start with the $p$-lower degree. Recall that the $p$-lower degree of a $q$-simplex $\sigma^{(q)}$ is the number of $q'$-simplices $\tau^{(q')}$ which are $p$-lower adjacent to $\sigma^{(q)}$, that is, those $q'$-simplices which contain a $p$-face $\gamma^{(p)}$ of $\sigma^{(q)}$. Hence, $\tau^{(q')}$ contributes to $\deg_L^p(\sigma^{(q)})$ as long as $\sig_L\big(\sigma^{(q)},\tau^{(q')};\gamma^{(p)}\big)$ does not vanish for some $\gamma^{(p)}$. Then $|\sig_L\big(\sigma^{(q)},\tau^{(q')};\gamma^{(p)}\big)|$ should be related to $\deg_L^p(\sigma^{(q)})$.

Following the proof of Theorem \ref{thm:Laplacian} is a straightforward computation to show that $\sig_L\big(\sigma^{(q)},\tau^{(q')};\gamma^{(p)}\big)$ can be given in terms of the entries of the matrices corresponding to $\partial$ and $\partial^*$ operators. However, as the common lower simplex might not be unique, the sum $\displaystyle\sum_{\gamma^{(p)}} |\sig_L\big(\sigma^{(q)},\tau^{(q')};\gamma^{(p)})|$ could be bigger than 1 and we would be counting the $q'$-simplex $\tau^{(q')}$ more than once when computing $\deg_L^p(\sigma^{(q)})$. 
\begin{definition}
Let $\sigma^{(q)}$ and $\tau^{(q')}$ be two simplices. We define the $p$-lower order of $\sigma^{(q)}$ and $\tau^{(q')}$ as the number $\ord^p_L(\sigma^{(q)},\tau^{(q')})$ of $p$-simplices of $K$ which are $p$-faces of both $\sigma^{(q)}$ and $\tau^{(q')}$. That is, 
$$\ord^p_L(\sigma^{(q)},\tau^{(q')})=\displaystyle\sum_{\gamma^{(p)}\in S_p(K)} |\sig_L\big(\sigma^{(q)},\tau^{(q')};\gamma^{(p)})|\,.$$

In that case, we shall say that $\sigma^{(q)}$ and $\tau^{(q')}$ are $p$-lower adjacent in order $\ord^p_L(\sigma^{(q)},\tau^{(q')})$. 
\end{definition}

This definition is used in the proof of the following statement.

\begin{theorem}\label{t:2}
Let $p$ and $q$ be non-negative integers. The $p$-lower degree of a $q$-simplex $\sigma_j^{(q)}$ is:
$$\deg_L^p(\sigma_j^{(q)})= -1+\sum_{q'=p}^{\dim K}\sum_k\min\big(1,\sum_{i}|b_{ij}^{(q,h)}||b^{(q',h')}_{ik}|\big)$$
with $h=q-p$ and  $h'=q'-p$. 
\end{theorem}
\begin{proof}
See Appendix \ref{s:A}.
\end{proof}

We shall state now an analogous formula to compute the $p$-upper degree. Let us recall that the $p$-upper degree of a $q$-simplex $\sigma^{(q)}$ is the number of $q'$-simplices $\tau^{(q')}$ which are $p$-upper adjacent to $\sigma^{(q)}$. In other words, those $\tau^{(q')}$ such that $|\sig_U\big(\sigma^{(q)}, \tau^{(q')};\gamma^{(p)}\big)|=1$, for some $p$-simplex $\gamma^{(p)}$. As in the lower degree setting, the common upper simplex $\gamma^{(p)}$ could be not unique, which motives the following definition.

\begin{definition}
Let $\sigma^{(q)}$ and $\tau^{(q')}$ be two simplices. The $p$-upper order of $\sigma^{(q)}$ and $\tau^{(q')}$, written $\ord^p_U(\sigma^{(q)},\tau^{(q')})$, is the number of $p$-simplices of which $\sigma^{(q)}$ and $\tau^{(q')}$ are faces. That is, 
$$\ord^p_U(\sigma^{(q)},\tau^{(q')})=\displaystyle\sum_{\gamma^{(p)}\in S_p(K)} |\sig_U\big(\sigma^{(q)},\tau^{(q')};\gamma^{(p)})|\,.$$
In that case, we shall say that $\sigma^{(q)}$ and $\tau^{(q')}$ are $p$-upper adjacent in order $\ord^p_U(\sigma^{(q)},\tau^{(q')})$. 
\end{definition}

The following result holds (and makes use of the definition above).

\begin{theorem}\label{t:3}
Let $p$ and $q$ be non-negative integers. The $p$-upper degree of a $q$-simplex $\sigma_j^{(q)}$ is:
$$\deg_U^p(\sigma_j^{(q)})= -1+\sum_{q'=0}^{p}\sum_k\min\big(1,\sum_{i}|b_{ji}^{(q+h,h)}||b_{ki}^{(q'+h',h')}|\big)$$
where $h=p-q$ and $h'=p-q'$.
\end{theorem}
\begin{proof}
See Appendix \ref{s:A}.
\end{proof}

We now state explicit formulas to compute de $p$-adjacency degree for a simplex and its maximal $p$-adjacency degree. Take $q,q'$ and $p$ non-negative integers, put $p'=q+q'-p$ and  fix basis of $C_{q'}(K),\, C_{q}(K)\,, C_{p}(K)\,,C_{p+1}(K)$ and $C_{p'}(K)$, namely,  
$\{\sigma_1^{(q')},\dots,\sigma_m^{(q')}\}$, $\{\sigma_1^{(q)},\dots,\sigma_n^{(q)}\}$, $\{\gamma_1^{(p)},\dots,\gamma_r^{(p)}\}$, $\{\gamma_1^{(p+1)},\dots,\gamma_s^{(p+1)}\}$ and $\{\tau_1^{(p')},\dots,\tau_t^{(p')}\}$, respectively, then the following formulas compute the adjacency degrees for simplices.

\begin{theorem}\label{t:4}
Let $q$ and $p$ be non-negative integers. Then:
\begin{enumerate}
\item
$$\deg^p_A(\sigma_j^{(q)})=\displaystyle\sum_{q'=p}^{\dim K}\displaystyle\sum_{k=1}^{f_{q'}} adj^p(\sigma_j^{(q)},\sigma_k^{(q')})$$
where $p'=q+q'-p$, $f_{q'}=\dim C_{q'}(K)$ and $adj^p(\sigma_j^{(q)},\sigma_k^{(q')})$ is defined in Appendix \ref{s:A}.
\item
$$
\deg^{p^*}_A(\sigma_j^{(q)})=\deg^p_A(\sigma_j^{(q)})-\displaystyle\sum_{q'=p}^{\dim K}\displaystyle\sum_{k=1}^{f_{q'}}\Delta_{q',k}
$$
with 
{\small $$
\Delta_{q',k}=\min\Big(1,\sum_{q'',\ell}|\sig_L(\sigma^{(q')}_k,\sigma_\ell^{(q'')};\sigma_k^{(q')})|\cdot adj^p(\sigma_j^{(q)},\sigma_\ell^{(q'')})\Big)adj^p(\sigma_j^{(q)},\sigma_k^{(q')})
$$}
where $p\leq q'\leq \dim K,\,\, 1\leq k\leq f_{q'},\, q'+1\leq q''\leq\dim K,\, 1\leq \ell\leq f_{q''}$ and $\{\sigma_\ell^{(q'')}\}_{\ell}$ is a basis of $C_{q''}(K)$.
\end{enumerate}
\end{theorem}

\begin{proof} See Appendix \ref{s:A}.
\end{proof}

\subsection{On some potential theoretical applications.}\label{ss:spec}\quad 

So far we have used the multi-combinatorial Laplacian of Definition \ref{d:multiLap} as a tool for computing the higher-order simplicial degrees of Section \ref{s:qhdeg}, but Laplacian operators have their own interest in many other research topics. The study of the eigenvalues of the graph Laplacian matrix, spectral graph theory, is essential for understanding different fields such as invariants in graph theory (\cite{C96,M91}), data analysis (\cite{BN03}), computer science (\cite{CS11}) or control theory (\cite{Bul18}). Due to its significance, spectral graph theory has been adapted to hypergraphs (see  \cite{C93}) and to simplicial complexes (see for example \cite{HJ13}), where it is also applied to the study of random walks (alternate sequences of classical boundary and coboundary operators in a simplicial complex) in \cite{MS13, PR17}. Recently, and using the language of sheaves coming from Algebraic Geometry, a new spectral sheaf theory has come onto the scene introducing the notion of a sheaf Laplacian $\Delta_{\mathcal F}$ (see \cite{HG18}), where if the chosen sheaf $\mathcal F$ on a simplicial complex is the constant sheaf, then one recovers the original combinatorial Laplacian. 

Most of the recent developments on spectral theory on simplicial complexes make use of a classical statement proved in \cite{E44}: the kernel of the combinatorial Laplacian computes the cohomology of the simplicial complex, that is, harmonic functions encodes the shape and many topological invariants of the simplicial network (an equivalent theorem is also proven in the sheaf context in \cite{HG18}). This result, and many of its implications, makes a significant use of the fact that the square of the classical boundary operator in a simplicial complex is zero: $\delta_q\circ \delta_{q+1}=0$.

The multi-parameter boundary operator and the multi-combinatorial Laplacian can be used to study similar applications in future research. For example, we can define a generalized random walk on a simplicial complex as an alternating sequence of the multi-parameter boundary $\delta_{q,h}$ and coboundary $\delta^*_{q,h}$ operators, a notion which recovers the usual definition of a random walk for $h=1$ (and which might also be related with the notion of $p$-walk given in \cite{HS19}). Moreover, since the multi-combinatorial Laplacian is a linear operator (so that we have a kernel), harmonic functions can be defined and thus we can study certain eigenvalues of the multi-combinatorial Laplacian. Nonetheless, the multi parameter boundary operator $\delta_{q,h}$ does not verify that its square is zero, that is, it is not always true that the composition $\delta_{q,h}\circ \delta_{q+h,h}$ is zero. Further research will be conducted in future works in order to elucidate what kind of topological invariants the multi-combinatorial Laplacian and its spectra are measuring in a simplicial network.

\section{Analysis of higher-order connectivity on real world networks}\label{s:realapp}

Little is known of higher-order connectivity models and patterns on complex networks since graph networks cannot encode higher-order relations and thus these multi-iterations are missed during data collection, and also due to the fact that storing simplicial complexes and working with them is normally a difficult computational task.

The main goal of this section is to prove the existence of a rich variety of distribution patterns for higher-order connectivity in real-world simplicial complex systems, which as far as we know, were not known. To achieve this, we propose the use of the maximal upper simplicial degree and the maximal simplicial degree (given in Definition \ref{d:simpdeg}) and to compute the size distribution and the simplicial degree distributions of 17 real-world datasets used in \cite{BASJK18}. These datasets are  already collected there as simplices bounded to a maximum of 25 nodes (which makes the calculations easier), and are obtained from real-world data of diverse domains such as coauthor networks, cosponsoring Congress bills, contacts in schools, drug abuse warning networks, e-mail networks, national drug code classes and substances, publications in online forums or users in online forums. The main results obtained from the analysis of these real-world datasets are the following:

\begin{itemize}
\item We find a rich variety of higher-order connectivity structures in simplicial complex networks.

\item We show that similar-type datasets reflect similar higher-order connectivity patterns.

\item We observe the same type of classical node and node-to-facets degree distributions -- a result that is consistent with the results of \cite{PGV17} and \cite{MDS15}.

\item We prove that for every analysed dataset, the degree distribution associated with the maximal upper simplicial degree follows  a power-law distribution similar to that of the corresponding distributions associated with the classical node and node-to-facets degree, but with a more pronounced decay (and thus less small degree saturation).

\item We show that if one instead uses the maximal simplicial degree, then its degree distribution is in general surprisingly different to that of the classical node degree one (and also to that of the node-to-facets and the maximal upper degree ones). In addition, a higher small degree saturation is shown.
\end{itemize}

In classical Network Science, hubs (which are nodes with a high number of connections) represent a deep organizing principle in the network topology. Similarly,  this structural analysis of higher-order connectivity patterns on real-world simplicial networks reveals the existence of ``simplicial hubs'' (simplices with a high number of simplicial connections) and their relevance in the study of the topology of the simplical network.

\subsection{Dataset description and summary statistics.}\label{s:datasets}

In order to carry out an applied analysis, let us start by describing and studying some statistical properties of the data processed. The data correspond to 17 real-world datasets which comprise sets of  simplices bounded to a maximum of 25 nodes, are publicly available at \url{https://github.com/arbenson/ScHoLP-Data} and have been obtained from \cite{BASJK18} (where Benson et al. analysed their temporal evolution, the simplicial closure phenomena and propose a higher-order link prediction). Their data and results have been extremely helpful for validating our own analysis.

Although the provided datasets contain more information, the following list summarizes the datasets in terms of the data that is relevant to our study, i.e., the datasets' simplices and nodes.

\begin{description}
	\item[coauth-DBLP, coauth-MAG-History, coauth-MAG-Geology.] A simplex is a publication. Nodes are the different authors.
	\item[congress-bills.] Nodes are members of Congress. A simplex is the set of members in a cosponsoring bill.
	\item[contact-high-school, contact-primary-school.] Nodes are people. Simplices are sets of people that were connected during 20-second intervals.
	\item[DAWN.] Simplices are the drugs used by a patient as obtained from the Drug Abuse Warning Network (DAWN). Nodes are the different drugs.
	\item[email-Enron, email-Eu.] A simplex is an email. The nodes of a simplex are the recipient addresses along with the sender one.
	\item[NDC-classes, NDC-substances.] Each simplex corresponds to a national drug code (NDC) for a drug. Nodes are the classes applied to that drug (NDC-classes) or the substances that make up that drug (NDC-substances).
	\item[tags-ask-ubuntu, tags-math-sx, tags-stack-overflow.] Each simplex is an entry/publication in an online forum. Nodes are the tags applied to it.
	\item[threads-stack-overflow, threads-math-sx, threads-ask-ubuntu.] Nodes are users of an online forum. A simplex is the set of users answering to the same question in the forum. 
\end{description}

\begin{table}[htb]
	\caption{Summary statistics.}
	\label{tab:sumstat}
	\resizebox{\textwidth}{!}{%
		\begin{tabular}{l rrrrrrr } \toprule
			& \begin{tabular}{c}nodes\end{tabular}
			& \begin{tabular}{c}simplices\end{tabular}
			& \begin{tabular}{c}distinct\\ simplices\end{tabular}
			& \begin{tabular}{c}unordered\\distinct\\ simplices\end{tabular}
			& \begin{tabular}{c}facets\end{tabular}
			& \begin{tabular}{c}\% facets / \\ simp\end{tabular}
			& \begin{tabular}{c}\% facets / \\ u.d.simp.\end{tabular}
			\\ \midrule
			coauth-DBLP & 1,924,991 & 3,700,067 & 2,599,087 & 2,466,792 & 1,730,664 & 46.77\% & 70.16\% \\ 
			coauth-MAG-Geology & 1,256,385 & 1,590,335 & 1,207,390 & 1,203,895 & 925,027 & 58.17\% & 76.84\% \\ 
			coauth-MAG-History & 1,014,734 & 1,812,511 & 895,668 & 895,439 & 774,495 & 42.73\% & 86.49\% \\ 
			congress-bills & 1,718 & 260,851 & 85,082 & 84,799 & 48,898 & 18.75\% & 57.66\% \\ 
			contact-high-school & 327 & 172,035 & 7,937 & 7,818 & 4,862 & 2.83\% & 62.19\% \\ 
			contact-primary-school & 242 & 106,879 & 12,799 & 12,704 & 8,010 & 7.49\% & 63.05\% \\ 
			DAWN & 2,558 & 2,272,433 & 143,523 & 141,087 & 72,421 & 3.19\% & 51.33\% \\ 
			email-Enron & 143 & 10,883 & 1,542 & 1,512 & 433 & 3.98\% & 28.64\% \\ 
			email-Eu & 998 & 234,760 & 25,791 & 25,027 & 8,102 & 3.45\% & 32.37\% \\ 
			NDC-classes & 1,161 & 49,724 & 1,222 & 1,088 & 563 & 1.13\% & 51.75\% \\ 
			NDC-substances & 5,311 & 112,405 & 10,025 & 9,906 & 6,555 & 5.83\% & 66.17\% \\ 
			tags-ask-ubuntu & 3,029 & 271,233 & 151,441 & 147,222 & 95,639 & 35.26\% & 64.96\% \\ 
			tags-math-sx & 1,629 & 822,059 & 174,933 & 170,476 & 96,596 & 11.75\% & 56.66\% \\ 
			tags-stack-overflow & 49,998 & 14,458,875 & 5,675,497 & 5,537,637 & 3,781,574 & 26.15\% & 68.29\% \\ 
			threads-ask-ubuntu & 125,602 & 192,947 & 167,001 & 166,999 & 149,025 & 77.24\% & 89.24\% \\ 
			threads-math-sx & 176,445 & 719,792 & 595,778 & 595,749 & 519,573 & 72.18\% & 87.21\% \\ 
			threads-stack-overflow & 2,675,955 & 11,305,343 & 9,705,709 & 9,705,562 & 8,694,667 & 76.91\% & 89.58\% \\
			\bottomrule
		\end{tabular}
	}
\end{table}

Table~\ref{tab:sumstat} shows the summary statistics obtained from the analysis of the datasets. Let us explain the terminology used and contrast the data against those collected in Table 1 of \cite{BASJK18}. Column `simplices' correspond to the data available in the datasets (called ``timestamped simplices'' in Table 1 of \cite{BASJK18}); column `distinct simplices' represents the number of simplices where duplicates are discarded (it corresponds to the column ``unique simplices'' in Table 1 of \cite{BASJK18}); the  `unordered distinct simplices' are counted by deleting the `distinct simplices' that have the same set of vertices (regardless of their order) than another distinct simplex; and the `facets' are calculated by removing the unordered distinct simplices that were already faces of another unordered distinct simplex.
\begin{example}
If a dataset were made of sets $\{1,4,6\}$, $\{4,1,6\}$, $\{6,4\}$ and $\{6,4\}$, the number of simplices is 4; the number of distinct simplices is 3 (simplex $\{6,4\}$ is repeated); the number of unordered distinct simplices is 2 ($\{1,4,6\}$ and $\{4,1,6\}$ contain the same set of vertices); and the number of facets is 1 ($\{6,4\}$ is a face of $\{1,4,6\}$).
\end{example}

Note that the lower the facets-to-simplices ratio in a category, the more repeated simplices of the same group of nodes there are in that category. For instance, with regard to a coauthors' network, this translates into a greater number of papers that the group of authors would have written jointly. Consequently, in  categories in which this ratio is very low, it would also be convenient to make an analysis of higher-order connectivity distributions taking into account these repetitions; that is to say, using a weighted simplicial degree. The weighted approach is out of the scope of this work but definitely appears as a future line of research. For this analysis we are not going to take into account the number of repetitions of a simplex in the datasets, nor the order of the vertices of a simplex, that is, we will be using the `unordered distinct simplices' and the facets columns of Table \ref{tab:sumstat} (in particular we are not considering the orientation).

\begin{table}[htb]
	\centering
	\caption{Facet size statistics. We show the maximum, average, median and most probable size of the facet.}
	\label{tab:facetstat}
	\resizebox{.55\textwidth}{!}{%
		\begin{tabular}{l rrrr} 
			\toprule
			& \begin{tabular}{c}max $s$\end{tabular}
			& \begin{tabular}{c}$\langle s \rangle$\end{tabular}
			& \begin{tabular}{c}$s_m$\end{tabular}
			& \begin{tabular}{c}$s_p$\end{tabular} \\ \midrule
			coauth-DBLP   &  25&3.51&3&3\\
			coauth-MAG-Geology & 25&3.50&3&2\\
			coauth-MAG-History   &   25&1.60&1&1\\
			congress-bills & 25&12.38&11&6\\
			contact-high-school  &5&2.42&2&2\\
			contact-primary-school &5&2.57&3&3\\
			DAWN  &16&4.90&4&4\\
			email-Enron &18&3.85&3&2\\
			email-Eu &25&3.82&2&2\\
			NDC-classes & 24&5.54&4&2\\
			NDC-substances & 25&6.70&5&1\\
			tags-ask-ubuntu &5&3.82&4&4\\
			tags-math-sx &5&3.96&4&4\\
			tags-stack-overflow &5&4.22&4&5\\
			threads-ask-ubuntu &14&2.00&2&2\\
			threads-math-sx &21&2.58&2&2\\
			threads-stack-overflow &25&2.59&2&2\\
			\bottomrule
		\end{tabular}
	}
\end{table}

\begin{figure}[htb]
	\centering
	\includegraphics[width=.7\linewidth]{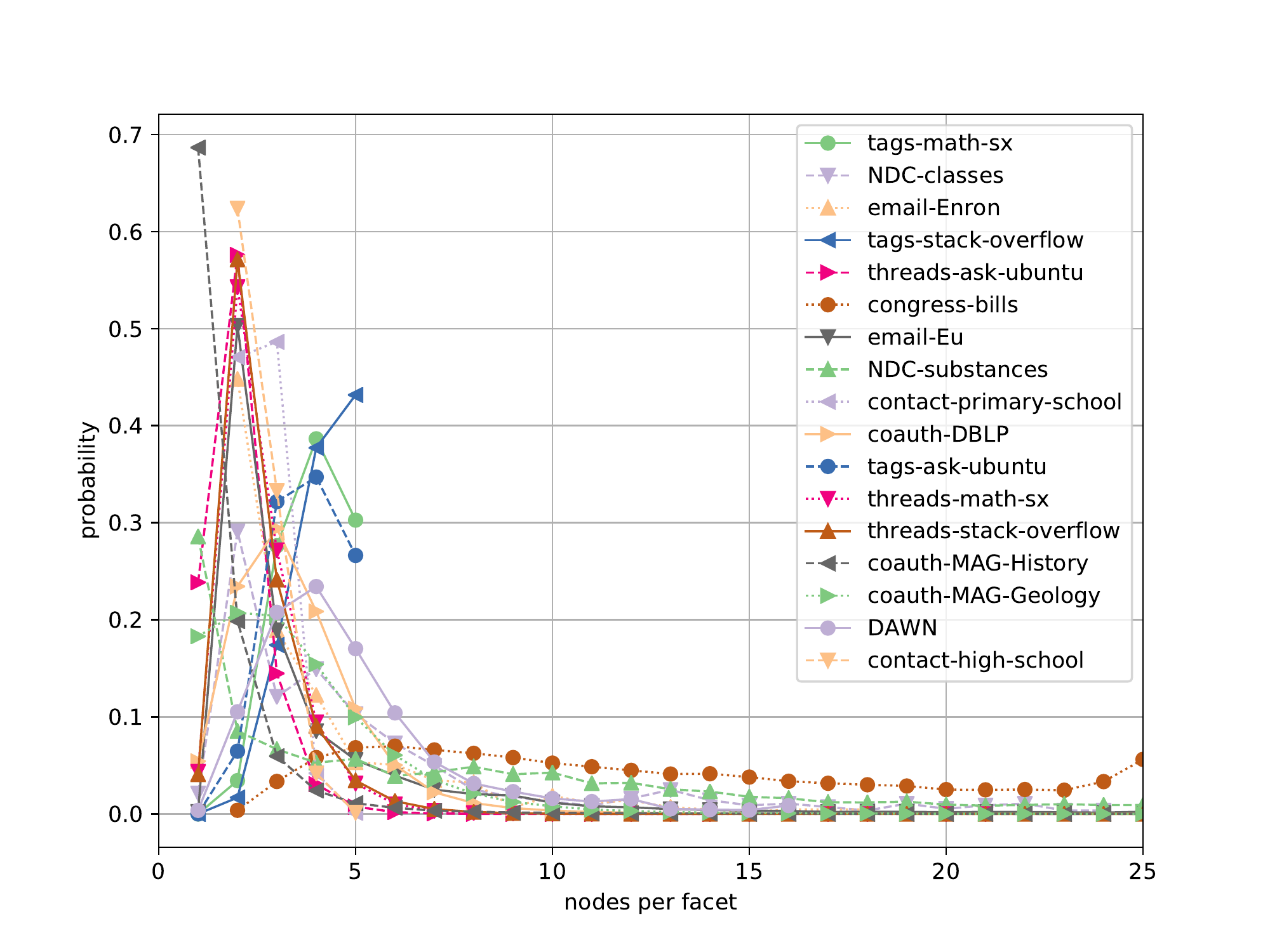}
	\caption{Distribution of facets size.}
	\label{fig:facetstat}
\end{figure}

As can be observed in Table \ref{tab:sumstat}, although the number of nodes and simplices of each category is very different, the facet size distributions are similar in almost all categories (see Figure \ref{fig:facetstat}). 

Table \ref{tab:facetstat} shows the facet size statistics. Recall that if a facet is given by a simplex $\sigma^{(q)}$, then its size is $s=q+1$, i.e. the number of vertices it has. We denote `max $s$' for the maximum size of the facets in a dataset; $\langle s \rangle$ stands for the average size of the facets in a dataset; we use $s_m$ for the median of the sizes of the facets; and $s_p$ denotes the most frequent size of the facets in the dataset, we will also refer to it as the most probable size of the facet. In this table, apart from the congress-bill dataset, we observe that both the median and the most probable value of the number of nodes per facet across categories is relatively low. Notice that the datasets' simplices were bounded to a maximum of 25 nodes, a value that eases the degrees' calculations and, as claimed in \cite{BASJK18}, only omits simplices with more than 25 nodes which ``are rare and not relevant for our analysis''. Verifying the real impact of this threshold is still work that should be performed, especially in the case of the congress-bill dataset, whose tail does not seem to tend to 0 when facet's size is 25. Moreover, as in classical network theory, this data incompleteness may also be responsible for an exaggerated pronounced decay on the distribution tail cutoffs we are about to present in the following subsection.

\subsection{Simplicial degree distributions.}

\begin{table}[t!]
	\caption{Classical-node degree and node-to-facets degree statistics. We show the maximum, the average and the median of both classical node degree (max $k$, $\langle k \rangle$ and $k_m$ resp.) and node-to-facets degree (max $k^F$, $\langle k^F \rangle$ and $k^F_m$ resp.)}
	\label{tab:classdeg}
	\centering
	\resizebox{.7\textwidth}{!}{%
		\begin{tabular}{@{}l rrr c rrr@{}}
			\toprule
			& \multicolumn{3}{c}{classical node degree} & \phantom{abc} & \multicolumn{3}{c}{node-to-facets degree} \\
			\cmidrule{2-4} \cmidrule{6-8}
			& max $k$ & $\langle k \rangle$ & $k_m$ && max $k^{F}$ & $\langle k^{F}\rangle$  & $k^{F}_m$ \\
			\midrule
			coauth-DBLP 			& 2,099 & 8.21 & 4 && 756 & 3.16 & 1 \\ 
			coauth-MAG-Geology 		& 1,236 & 8.15 & 4 && 545 & 2.57 & 1 \\ 
			coauth-MAG-History 		& 892 & 2.28 & 0 && 401 & 1.22 & 1 \\ 
			congress-bills 			& 3 & 0.00 & 0 && 2,562 & 352.40 & 231 \\ 
			contact-high-school 	& 87 & 35.58 & 36 && 90 & 35.99 & 35 \\ 
			contact-primary-school 	& 134 & 68.74 & 68 && 174 & 85.19 & 82 \\ 
			email-Enron 			& 69 & 25.17 & 23 && 36 & 11.65 & 10 \\ 
			email-Eu 				& 394 & 58.72 & 45 && 358 & 31.05 & 17 \\ 
			DAWN 					& 1,241 & 96.14 & 20 && 13,204 & 138.70 & 8 \\ 
			NDC-classes 			& 167 & 10.72 & 4 && 66 & 2.69 & 1 \\ 
			NDC-substances 			& 903 & 33.24 & 5 && 519 & 8.26 & 1 \\ 
			tags-ask-ubuntu 		& 2,082 & 87.62 & 37 && 9,093 & 120.46 & 22 \\ 
			tags-math-sx 			& 1,108 & 112.57 & 62 && 8,627 & 234.68 & 53 \\ 
			tags-stack-overflow 	& 4 & 0.27 & 0 && 375,325 & 319.47 & 24 \\ 
			threads-ask-ubuntu 		& 2,727 & 2.98 & 1 && 2,054 & 2.37 & 1 \\ 
			threads-math-sx 		& 11,001 & 12.35 & 2 && 9,996 & 7.61 & 1 \\ 
			threads-stack-overflow 	& 41,474 & 15.70 & 2 && 32,358 & 8.41 & 1 \\
			\bottomrule
		\end{tabular}
	}
\end{table}

The second part of the topological analysis of the higher-order connectivity in simplicial data systems corresponds to the study of the distributions associated with the maximal upper simplicial degree and the maximal simplicial degree (given in Definition \ref{d:simpdeg}) of the 17 real-world dataset mentioned before (see Tables \ref{tab:maxsimpupdeg} and \ref{tab:maxsimpdeg} for degree statistics). We will compare  them with the simplicial degree distributions of the classical node degree (\cite{BA16}) and the node-to-facets degree of \cite{PGV17} (see Table \ref{tab:classdeg} for degree statistics) and interpret the results and figures by contrasting them with known properties of classical node degree distributions in Network Theory (\cite{BA16, BR02}). These simplicial degree distributions are studied both for the median and the most probable value of the number of nodes per facet since they are statistically representative values which, being relatively low (see Table \ref{tab:maxsimpupdeg} for instance), facilitate the computational computations.

In Tables \ref{tab:classdeg}, \ref{tab:maxsimpupdeg} and \ref{tab:maxsimpdeg} we have abbreviated the notations as follows: 
\begin{itemize}
	\item $k$ denotes the classical degree of a node $v$ in a network: the number of its incident edges. That is, $k=\deg (v)=\deg_U^{1,1}(v)$ following Definition \ref{d:hpUdeg}.
	\item $q_m=s_m-1$ and recall that $s_m$ stands for the median number of nodes per facet in each dataset.
	\item $q_p=s_p-1$ and $s_p$ is the most frequent size of the facets of the dataset.
	\item We denote $k_U^{*}(q)$ the maximal upper simplicial degree of a $q$-simplex $\sigma^{(q)}$: the number of different facets to which $\sigma^{(q)}$ belongs to. That is, $k_U^{*}(q)=\deg_U^*(\sigma^{(q)})$ following Definition \ref{d:simpdeg}.
	\item We denote $k^{F}$ the node to facets degree: the number of different maximal simplices to which a node $v$ belongs to. That is, $k^F=k_U^*(0)=\deg_U^*(v)$.
	\item We denote $k^{*}(q)$ the maximal simplicial degree of a $q$-simplex $\sigma^{(q)}$: the number of different facets to which $\sigma^{(q)}$ belongs to and to which its strict faces also belong to. That is, $k^{*}(q)=\deg_A^*(\sigma^{(q)})$ following Definition \ref{d:simpdeg}.
\end{itemize}

\begin{table}[htb]
	\caption{Maximal upper simplicial degree statistics. We show the maximum, average and median value of the maximal upper simplicial degree for the median $q_m=s_m - 1$ and most probable $q_p=s_p - 1$ simplex' dimension.}
	\label{tab:maxsimpupdeg}
	\resizebox{\textwidth}{!}{%
		\begin{tabular}{l rrrr c rrrr}
			\toprule
			& $q_m$ & max $k_U^*(q_m)$ & $\langle k_U^*(q_m)\rangle$ & $k_U^*(q_m)_m$
			& \phantom{ab}
			& $q_p$ & max  $k_U^*(q_p)$ & $\langle k_U^*(q_p)\rangle$ & $k_U^*(q_p)_m$ \\ 
			\cmidrule{2-5} \cmidrule{7-10}
			coauth-DBLP 			& 2 & 93 & 1.11 & 1 && 2 & 93 & 1.11 & 1 \\ 
			coauth-MAG-Geology 		& 2 & 60 & 1.12 & 1 && 1 & 141 & 1.29 & 1 \\
			coauth-MAG-History 		& 0 & 401 & 1.22 & 1 && 0 & 401 & 1.22 & 1 \\
			contact-high-school 	& 1 & 19 & 1.58 & 1 && 1 & 19 & 1.58 & 1 \\
			contact-primary-school 	& 2 & 6 & 1.04 & 1 && 2 & 6 & 1.04 & 1 \\
			DAWN 					& 3 & 197 & 1.15 & 1 && 3 & 197 & 1.15 & 1 \\
			email-Enron 			& 2 & 12 & 1.47 & 1 && 1 & 18 & 2.25 & 1 \\
			email-Eu 				& 1 & 95 & 3.09 & 1 && 1 & 95 & 3.09 & 1 \\
			NDC-classes 			& 3 & 45 & 1.55 & 1 && 1 & 63 & 2.19 & 1 \\
			NDC-substances 			& 4 & 50 & 1.36 & 1 && 0 & 519 & 8.26 & 1 \\
			tags-ask-ubuntu 		& 3 & 59 & 1.10 & 1 && 3 & 59 & 1.10 & 1 \\
			tags-math-sx 			& 3 & 59 & 1.18 & 1 && 3 & 59 & 1.18 & 1 \\
			threads-ask-ubuntu 		& 1 & 81 & 1.05 & 1 && 1 & 81 & 1.05 & 1 \\
			threads-math-sx 		& 1 & 330 & 1.21 & 1 && 1 & 330 & 1.21 & 1 \\
			\bottomrule
		\end{tabular}
	}
\end{table}

\begin{table}[htb]
	\caption{Maximal simplicial degree statistics. We show the maximum, average and median value of the maximal simplicial degree for the median $q_m=s_m - 1$ and most probable $q_p=s_p - 1$ simplex' dimension.}
	\label{tab:maxsimpdeg}
	\resizebox{\textwidth}{!}{%
		\begin{tabular}{l rrrr c rrrr}
			\toprule
			& $q_m$ & max $k^*(q_m)$ & $\langle k^*(q_m)\rangle$ & $k^*(q_m)_m$
			& \phantom{ab}
			& $q_p$ & max  $k^*(q_p)$ & $\langle k^*(q_p)\rangle$ & $k^*(q_p)_m$ \\
			\cmidrule{2-5} \cmidrule{7-10}
			coauth-DBLP				& 2 & 2,043 & 64.64 & 30 		&& 2 & 2,043 & 64.64 & 30 \\ 
			coauth-MAG-Geology 		& 2 & 934 & 57.38 & 36 			&& 1 & 844 & 31.32 & 15 \\ 
			coauth-MAG-History 		& 0 & 401 & 1.22 & 1 			&& 0 & 401 & 1.22 & 1 \\ 
			contact-high-school 	& 1 & 168 & 81.76 & 80 			&& 1 & 168 & 81.76 & 80 \\ 
			contact-primary-school 	& 2 & 477 & 293.68 & 293 		&& 2 & 477 & 293.68 & 293 \\ 
			DAWN 					& 3 & 26,329 & 6,869.43 & 6,519	&& 3 & 26,329 & 6,869.43 & 6,519 \\ 
			email-Enron 			& 2 & 79 & 41.78 & 42 			&& 1 & 64 & 27.68 & 27 \\ 
			email-Eu 				& 1 & 636 & 132.40 & 107 		&& 1 & 636 & 132.40 & 107 \\ 
			NDC-classes 			& 3 & 100 & 48.34 & 45 			&& 1 & 88 & 19.93 & 15 \\ 
			NDC-substances 			& 4 & 1,660 & 486.90 & 457 		&& 0 & 519 & 8.26 & 1 \\ 
			tags-ask-ubuntu 		& 3 & 26,946 & 6,612.21 & 6,076 && 3 & 26,946 & 6,612.21 & 6,076 \\ 
			tags-math-sx 			& 3 & 27,123 & 7,859.79 & 7,245 && 3 & 27,123 & 7,859.79 & 7,245 \\ 
			threads-ask-ubuntu 		& 1 & 3,651 & 286.41 & 88 		&& 1 & 3,651 & 286.41 & 88 \\ 
			threads-math-sx 		& 1 & 18,625 & 1,334.60 & 582 	&& 1 & 18,625 & 1,334.60 & 582 \\
			\bottomrule
		\end{tabular}
	}
\end{table}

\begin{figure}[htb]
	\centering
	\includegraphics[width=0.48\linewidth]{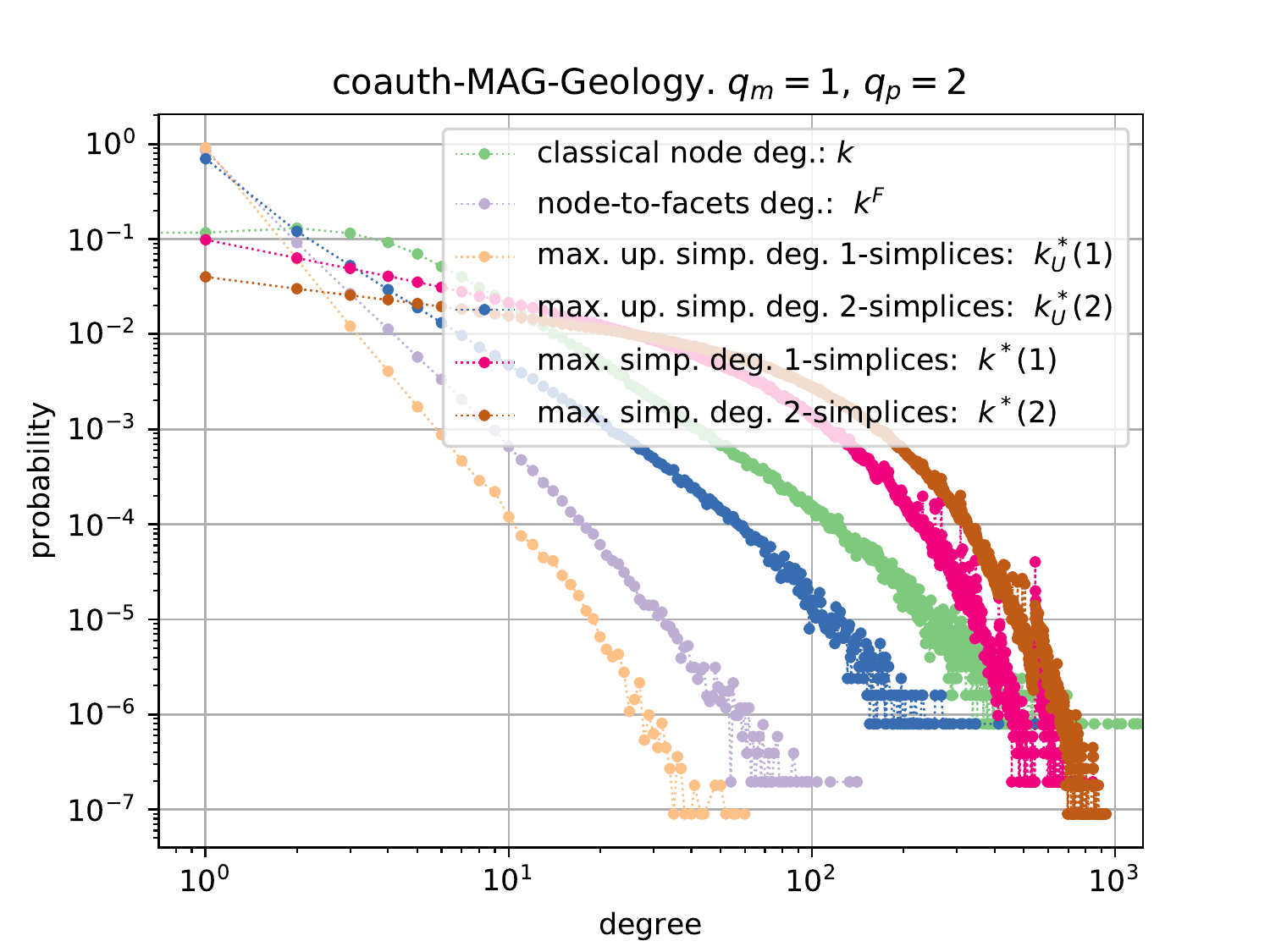}
	\includegraphics[width=0.48\linewidth]{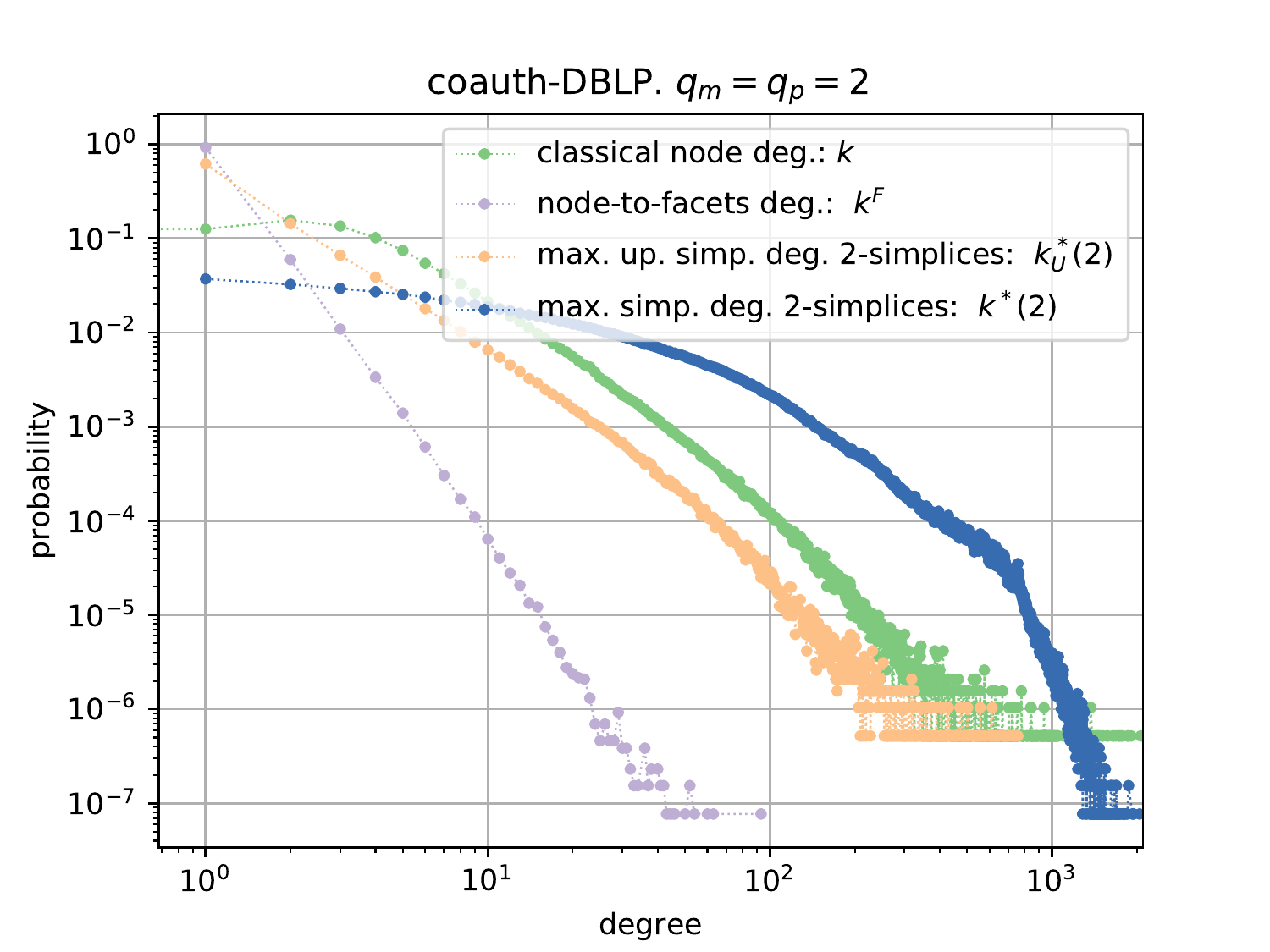}
	\caption{Log-log plot of simplicial degree distributions of the coauthor's datasets.}
	\label{fig:coauthor_llmsd}
\end{figure}

\begin{figure}[htb]
	\centering
	\includegraphics[width=0.48\linewidth]{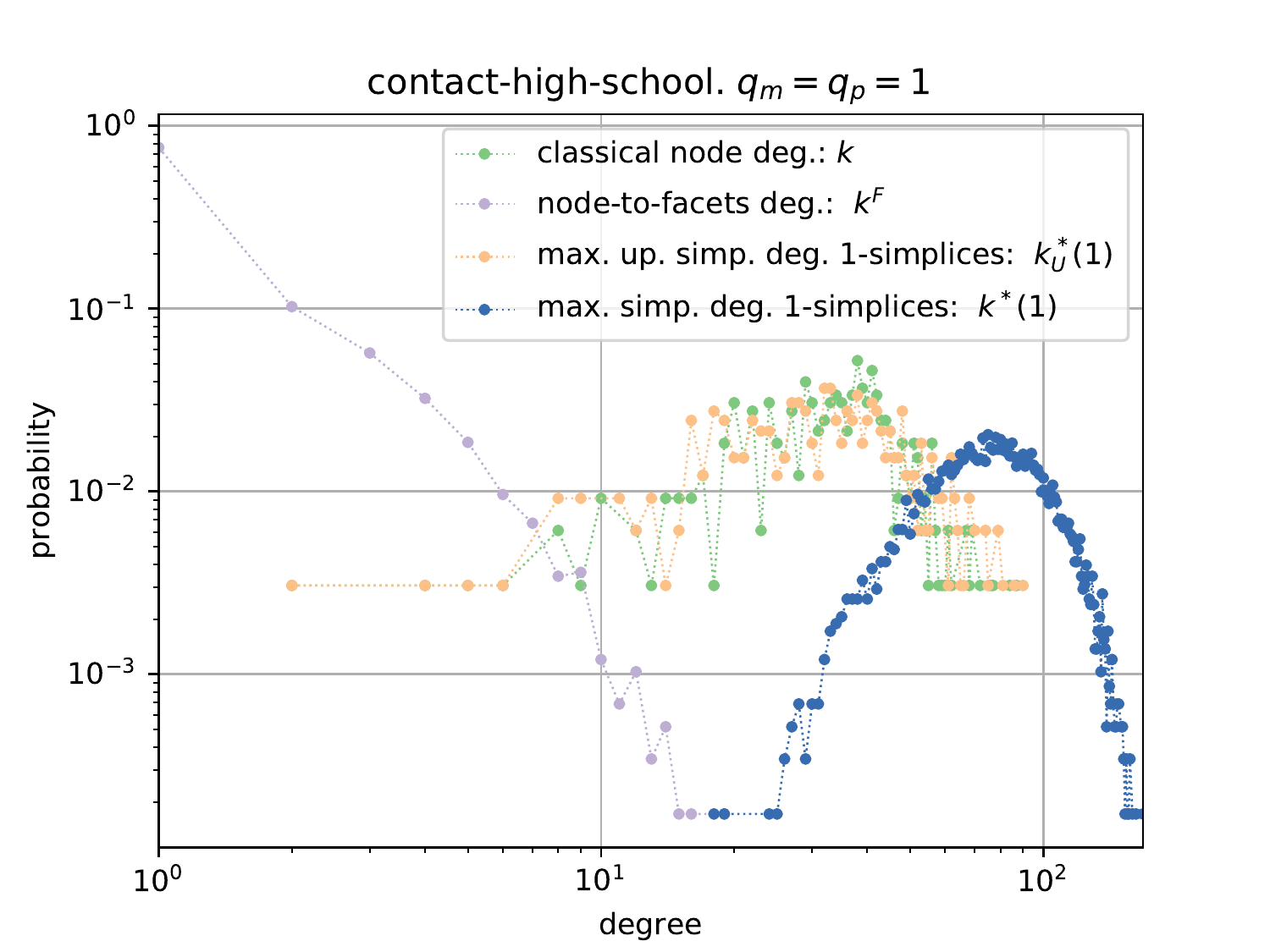}
	\includegraphics[width=0.48\linewidth]{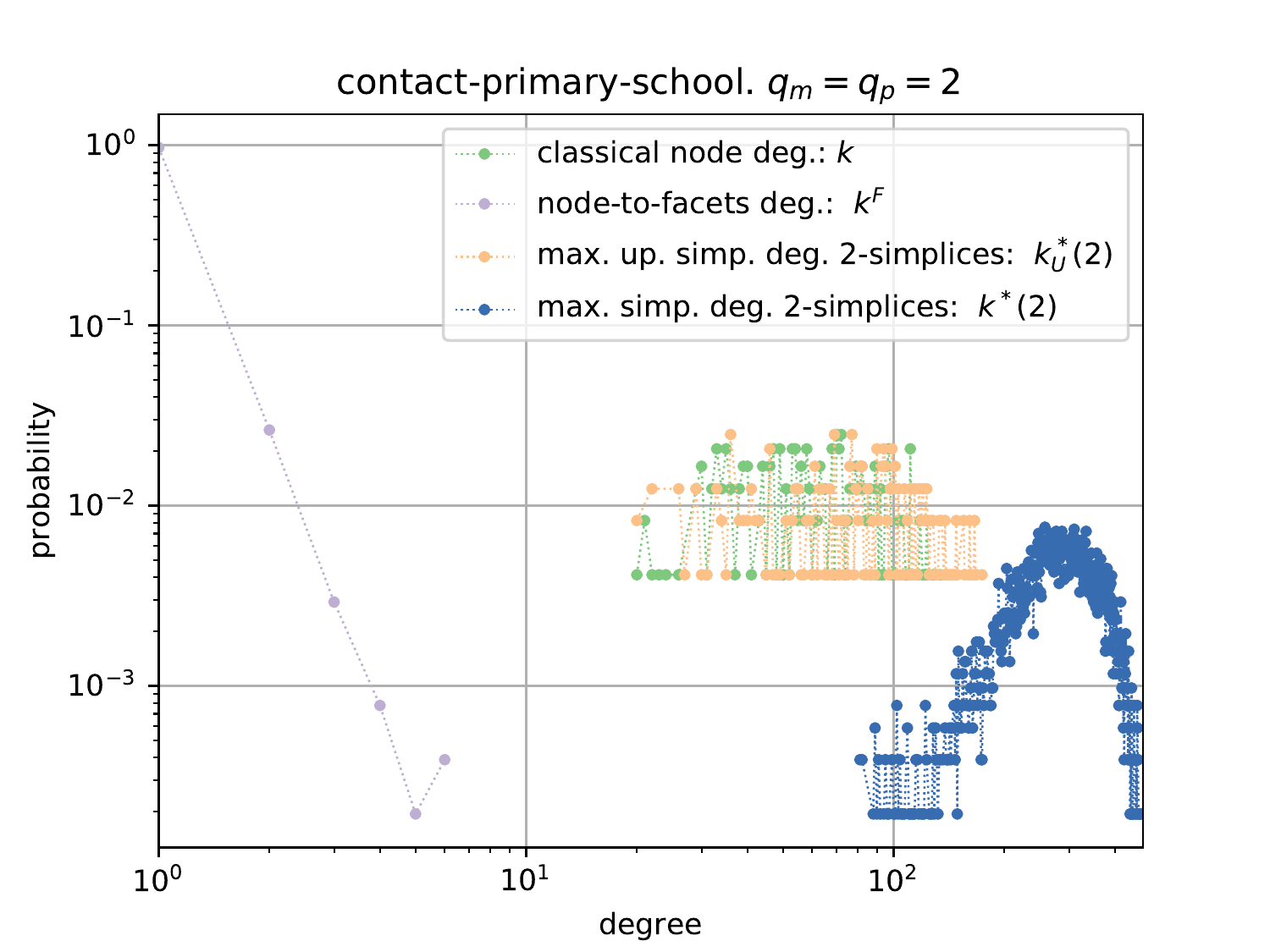}
	\caption{Log-log plot of simplicial degree distributions of contact-high-school and contact-primary-school datasets.}
	\label{fig:G2}
\end{figure} 

\begin{figure}[htb]
	\centering
	\includegraphics[width=0.48\linewidth]{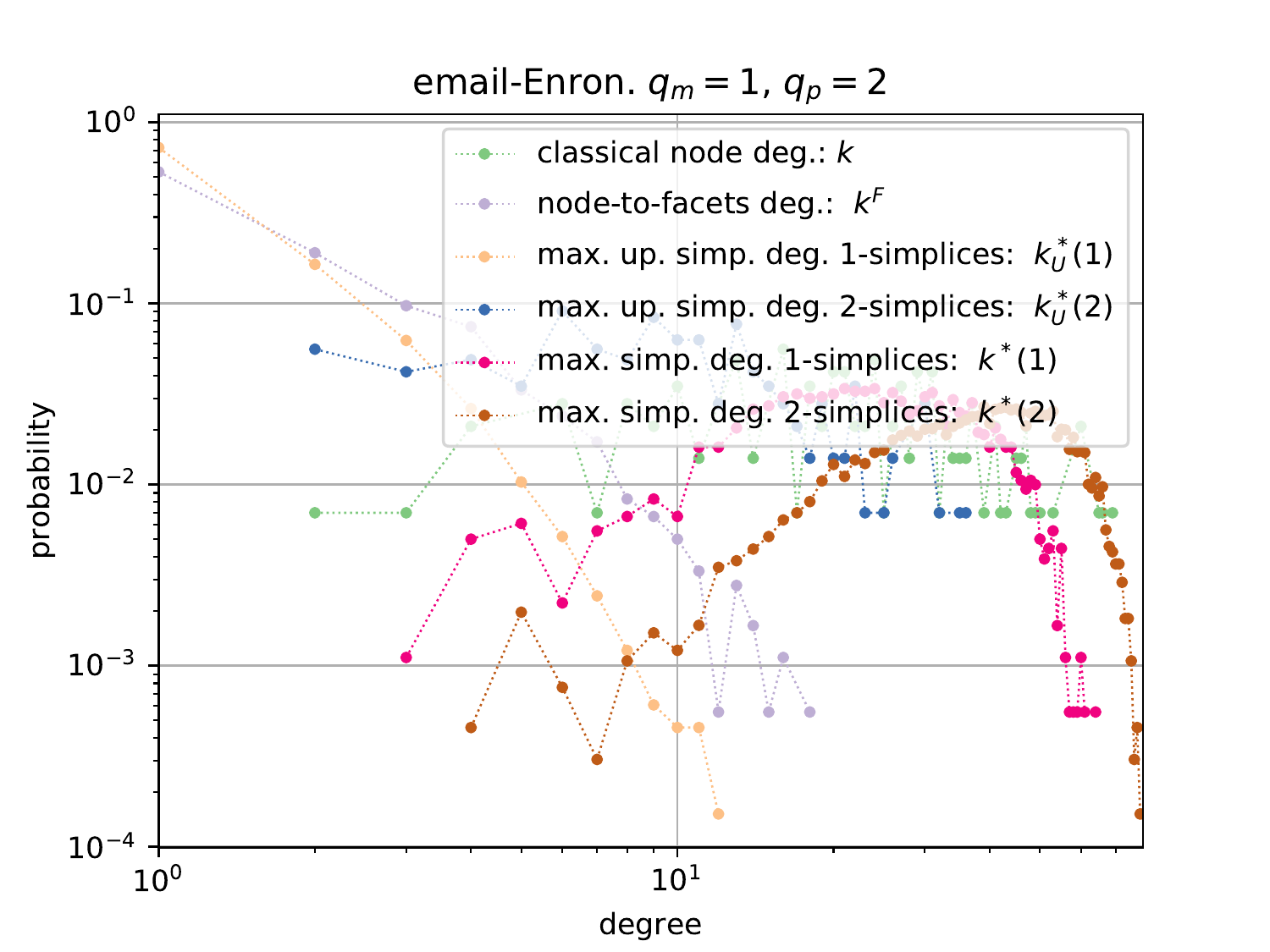}
	\includegraphics[width=0.48\linewidth]{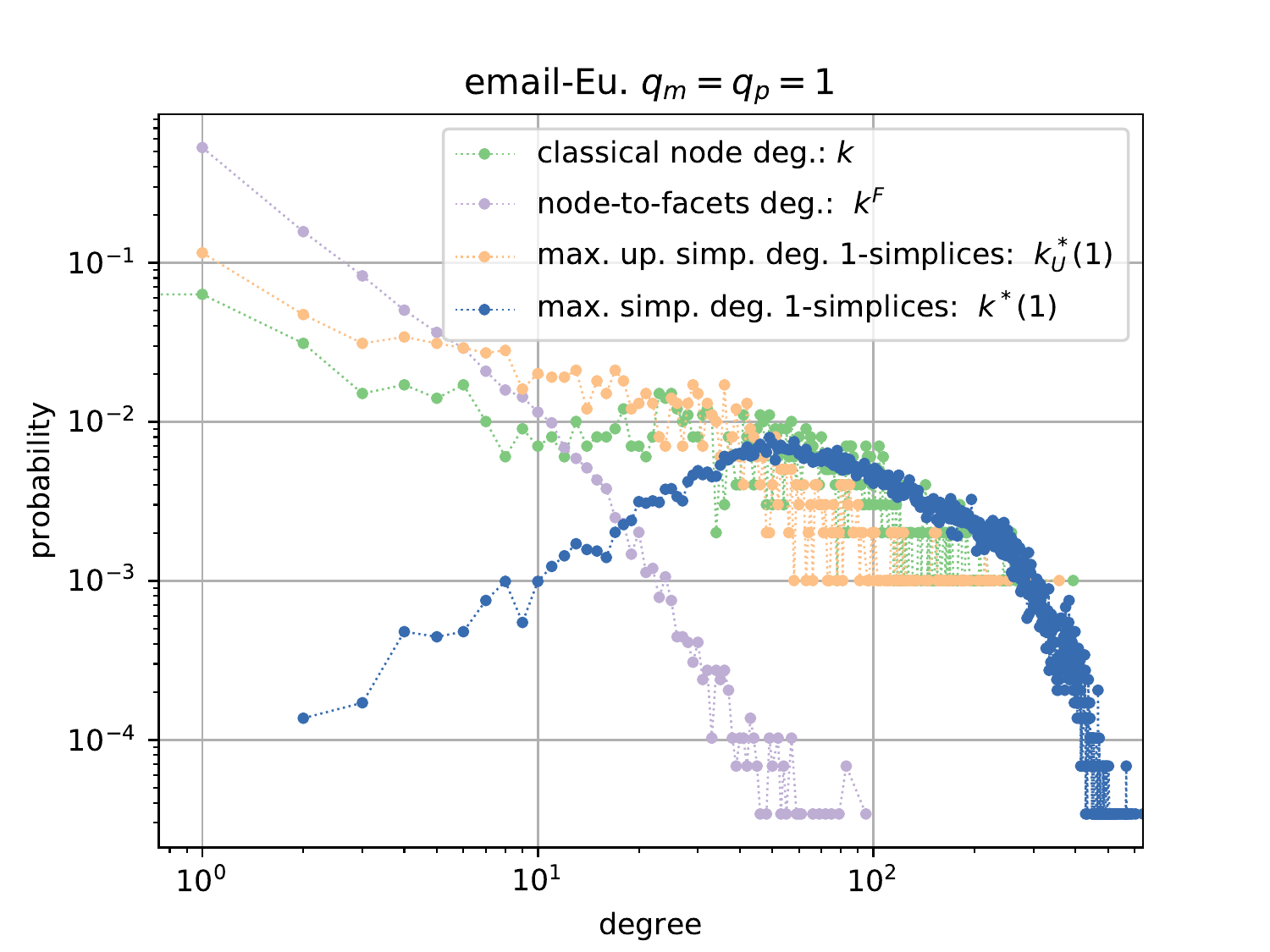}
	\caption{Log-log plot of simplicial degree distributions of emails' datasets.}
	\label{fig:G3a}
\end{figure} 

\begin{figure}[htb]
	\centering
	\includegraphics[width=0.48\linewidth]{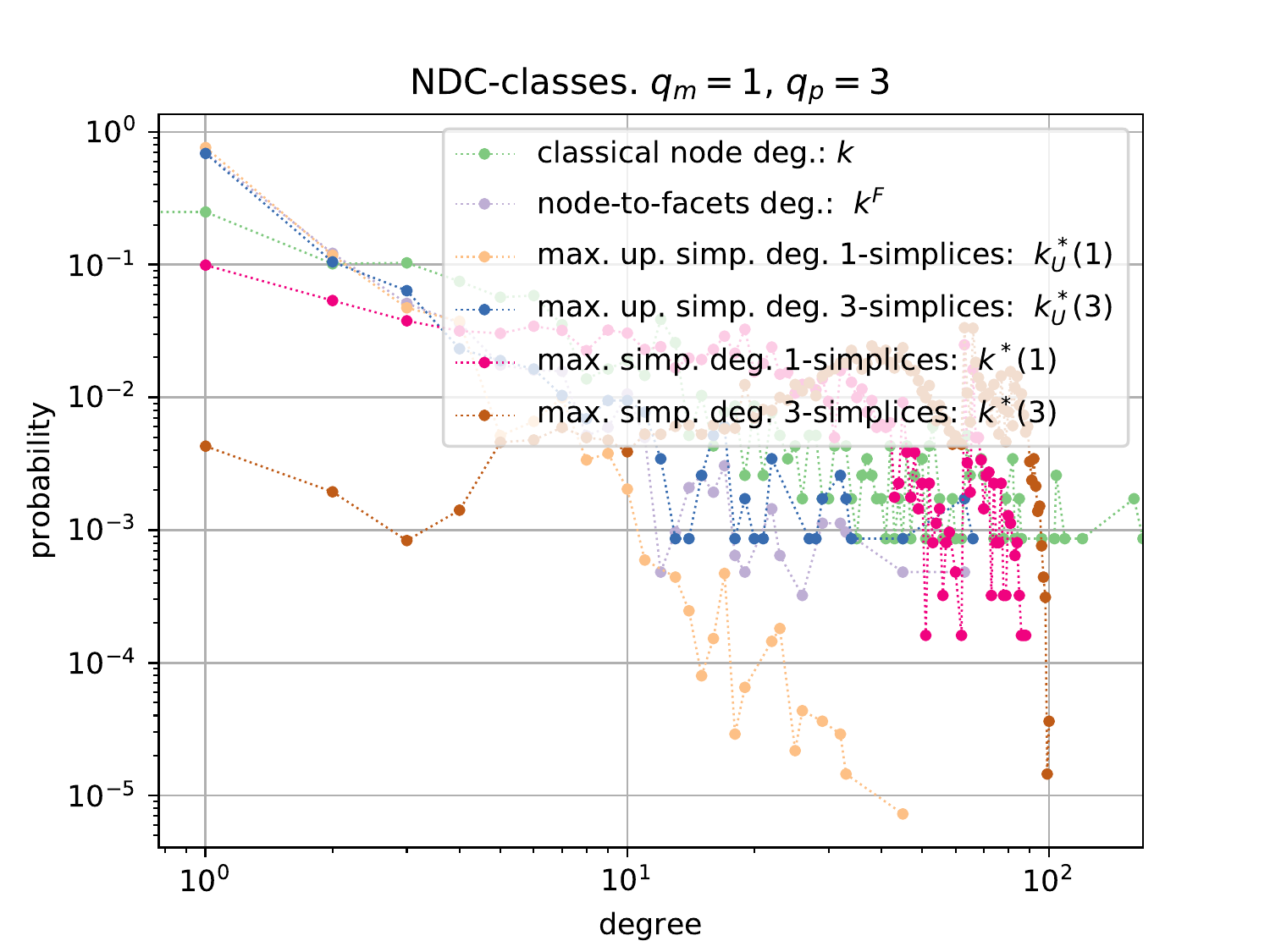}
	\includegraphics[width=0.48\linewidth]{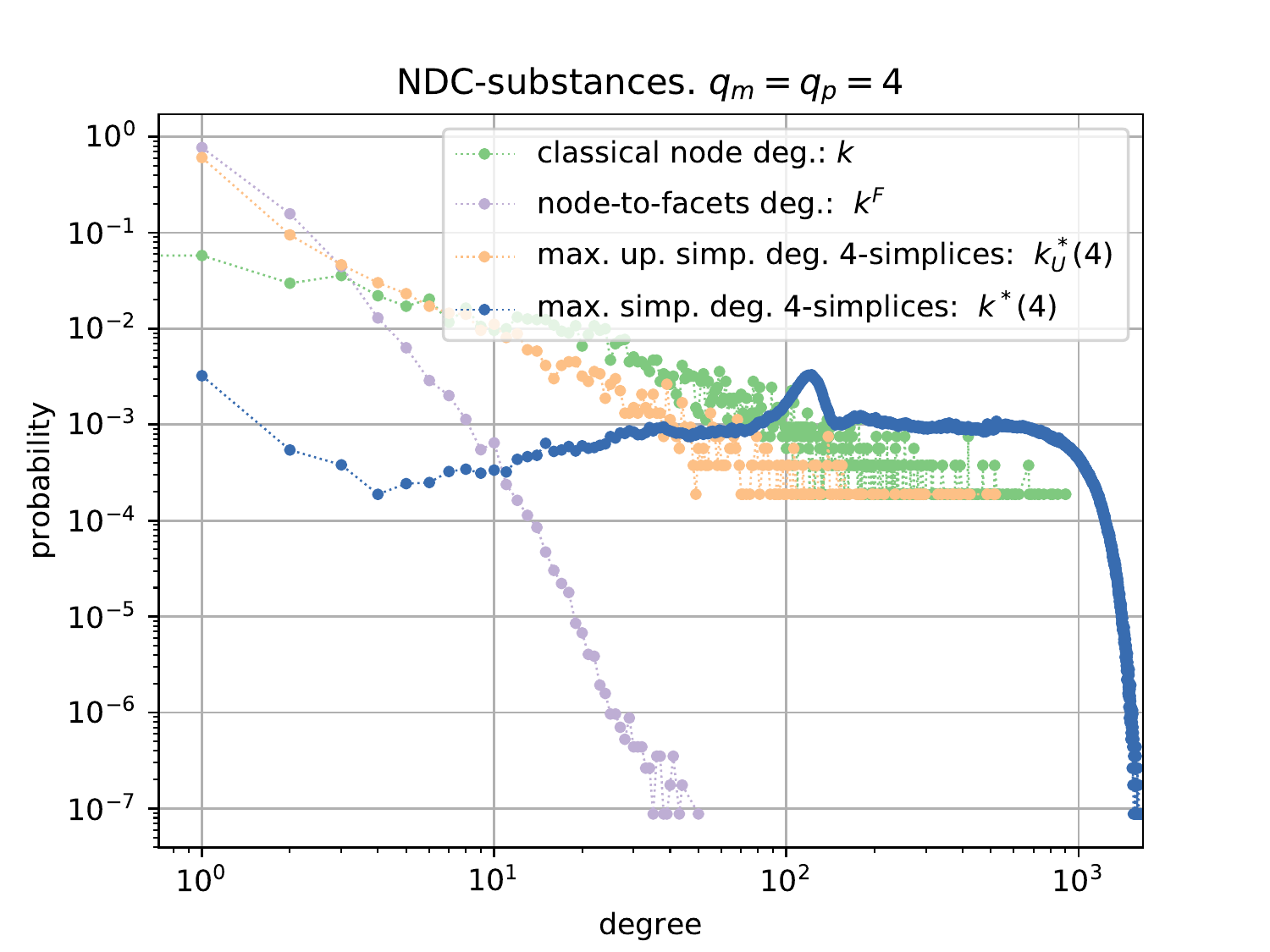}
	\caption{Log-log plot of simplicial degree distributions of NDC datasets.}
	\label{fig:G3b}
\end{figure}

Let us start by stating common conclusions to the degree distributions represented in Figures \ref{fig:coauthor_llmsd}, \ref{fig:G2}, \ref{fig:G3a}, \ref{fig:G3b}, \ref{fig:G4} and \ref{fig:G5}. 

First thing to notice is that a rich variety of higher-order connectivity patterns is unveiled and that the same type of classical-node and node-to-facets degree distributions is observed in all datasets -- a result that is consistent with the results of \cite{PGV17} and \cite{MDS15}. In addition to this, a consistency of our data treatment might be observed by taking into account that, all the real-world datasets analysed having the same type of classical node or node-to-facets degree distributions (grouped in the figures) have similar and consistent higher-order connectivity distribution patterns for both the maximal upper simplicial degree and the maximal simplicial degree.

With regard to the distribution associated with the maximal upper simplicial degree, the figures indicate that this distribution is always closer to a power law distribution, even for those datasets whose node degree distribution does not follow a power law, and it has a more pronounced decay for those datasets whose node degree distribution also follows a power law. Figures show a stretched exponential cutoff and less small $k^*_U$-degree saturation, even for the real-world datasets whose nature was not that of a scale-free network. This, together with the fact that the median of the maximal upper degree is always $k^*_U(q_m)_m=k^*_U(q_p)_m=1$ (see Table \ref{tab:maxsimpupdeg}), might suggest that for the simplicial maximal upper adjacency there are less simplicial communities concentrating a huge amount of maximal upper simplicial collaborations, which could be interpreted as there are fewer and smaller simplicial hubs (in this case maximal upper simplicial hubs).

Let us now comment on some results common to all datasets for the distributions associated with the maximal simplicial degree of Definition \ref{d:simpdeg}. We believe that this simplicial degree is a better notion to measure the relevance of a simplicial community, since it also keeps track of the connections of the strict faces or sub-communities of a simplex (which would lead to a right notion of ``simplicial influencers'' or ``influential communities'').

\begin{itemize}
\item We first observe a rich variety of higher-order connectivity patterns and, in many figures, a surprisingly different shape to that of the classical node degree distribution. It can be seen that all the figures present either a higher small $k^*$-degree saturation than in the classical node $k$-degree distributions, or are directly of a bell-shaped curve type. As these characteristics are typical of networks with a certain amount of random behaviour, this last fact might suggest that when we use of the maximal simplicial degree a random component is added to the simplicial network, making its topology more homogeneous. This remarkable phenomenon may also be associated with the fact that simplicial communities potentiate a supportive behaviour and add an initial attractiveness (working as a team, and particularly as a ``team of teams'',  always benefits and also makes the team network more solidary than the individual network). 

\item In most figures one can also see that a pronounced tail cutoff exists for the $k^*$-degree distribution, which resembles to that of a sublinear preferential attachment model but having a more pronounced decrease than that of the classical node degree distribution. Moreover, the $k^*$-degree distribution takes smaller values than those taken in the classical node $k$-degree distribution, which might be due to the fact that the simplicial sample size is usually larger than the number of nodes. These facts lead to fewer and smaller simplicial hubs and also to a slower growth of the maximum $k^*$-degree. We believe that this should also imply that there are longer distances in the simplicial network (see \cite{HS19} for some definitions on distances on simplicial complexes), and that the topologies associated with such networks deviate from those of small world networks.
\end{itemize}

Nonetheless, let us point out that, as in Network Science, the structural cutoffs mentioned above might be a consequence of data incompleteness (recall that simplices of the datasets where bound to a maximum of $25$ nodes), so that this additional phenomena still need to be understood. In fact, in many systems empirical data is not enough to properly fit real-world degree distributions and thus generative simplicial models, which would analytically predict the expected simplicial degree distributions, would be necessary. From this point of view, it would be quite interesting to have an equivalent notion to that of the degree exponent $\gamma$ of a distribution in our simplicial case, and to study its relationship with or dependence on the classical degree exponent (which would shed some light on classifying simplicial networks). If we denote by $\gamma^F$, $\gamma^*_U$ and $\gamma^*$ the presumably existing degree exponents associated with the node-to-facets degree, the maximal upper simplicial degree and the maximal simplicial degree, respectively, Figures \ref{fig:coauthor_llmsd}, \ref{fig:G2}, \ref{fig:G3a}, \ref{fig:G3b} and \ref{fig:G4} also suggest that for the tail cutoff of the distributions one would have $\gamma^*\geq \gamma^*_U\geq \gamma^F \geq\gamma$.

Below we will comment on particular characteristics and interpretations that can be deduced from this analysis for the different groups of real-world datasets represented in each of the figures.

Figure  \ref{fig:coauthor_llmsd} (coauth-MAG-Geology and coauth-DBLP datasets) shows in the low values of the maximal simplicial degree axis, a higher $k^*$-degree saturation, which causes an approximation of the graphic (in this region) to a stretched exponential curve distribution. In the high $k^*$-degree region, the structural cutoff resembles to a sublinear preferential attachment regime but with a faster decay than that of the classical node $k$-degree distribution. This might be interpreted as less dependence on simplicial network growth and bigger initial simplicial attractiveness or fitness. The probability of being part of a simplicial community with a relatively high $k^*$-degree is greater than that of being an individual with a high classical node $k$-degree, whereas finding a maximal simplicial hub is less probable than finding a usual hub.

In Figure \ref{fig:G2} (contact-high-school and contact-primary-school datasets) we observe a random model for the classical $k$-degree and the node-to-facets $k^F$-degree distributions; a fact that also holds true, and in a more clear way, for the maximal simplicial $k_U^*$-degree distribution. Nonetheless, the maximal upper simplicial $k^*$-degree distribution is close to a power law model; a strange phenomenon that still needs to be understood. Something similar is happening in Figure \ref{fig:G3a} for the e-mail-Enron and email-EU datasets.

\begin{figure}[htb]
	\centering
	\includegraphics[width=0.48\linewidth]{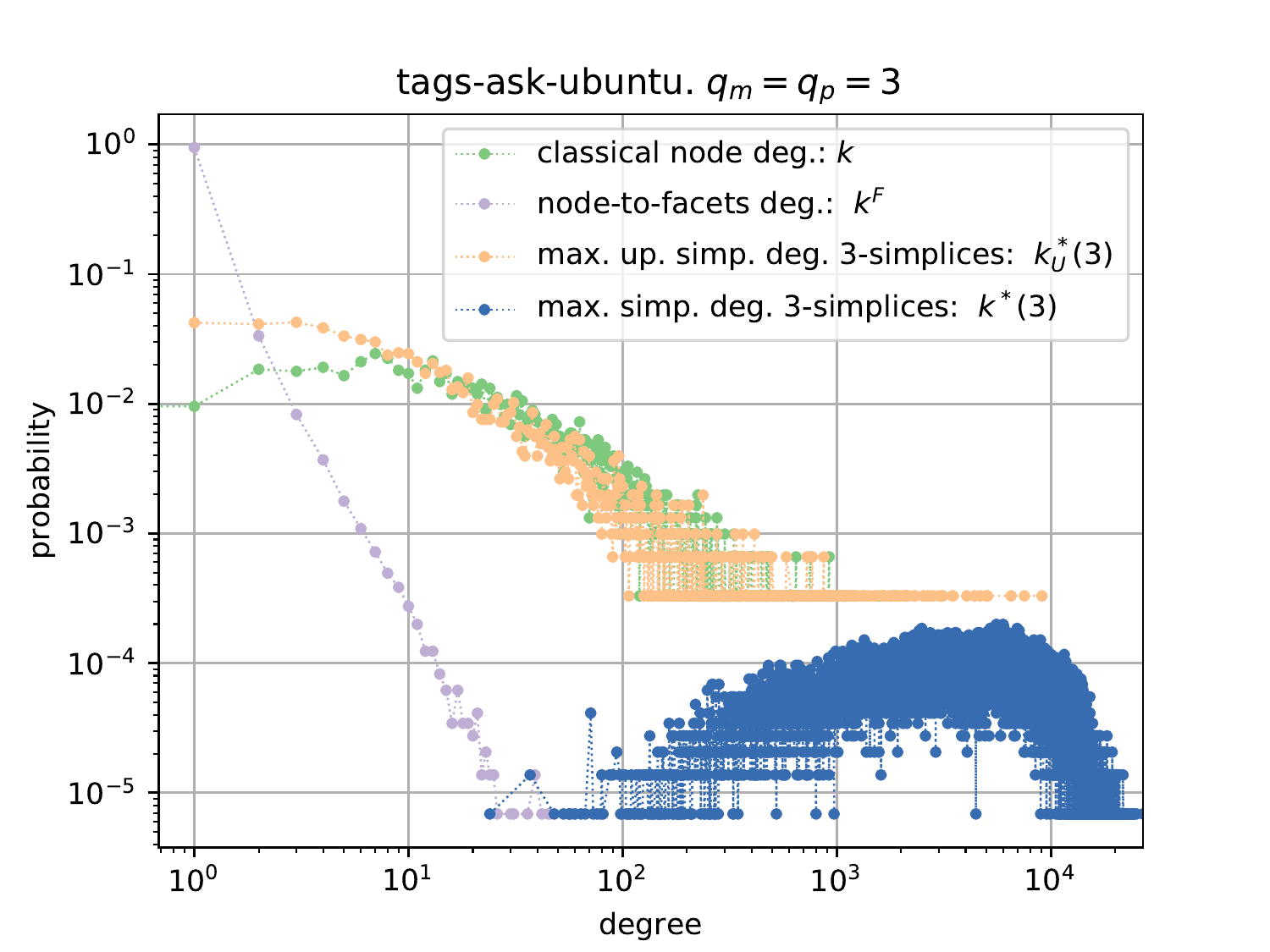}
	\includegraphics[width=0.48\linewidth]{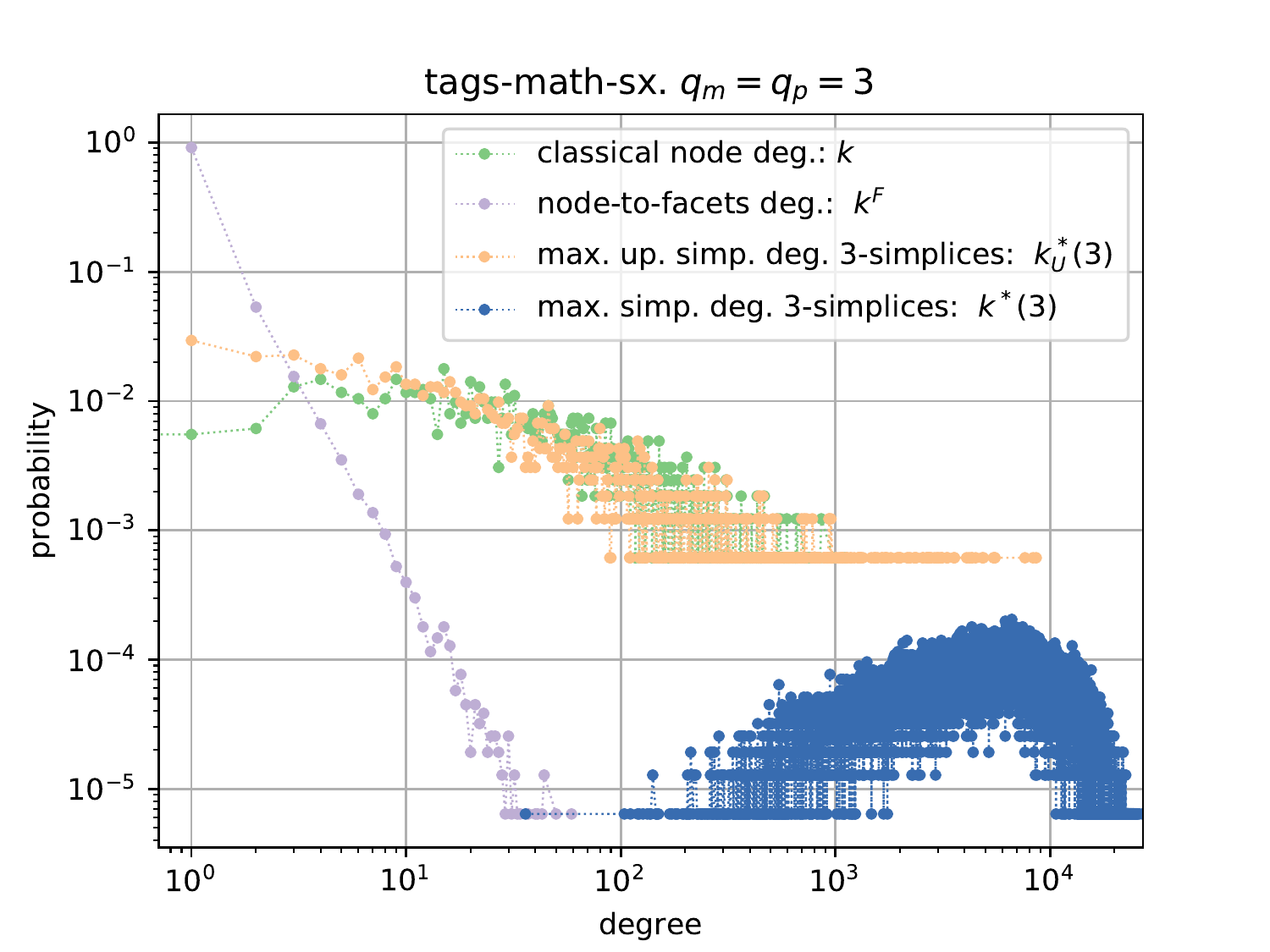}

	\includegraphics[width=0.48\linewidth]{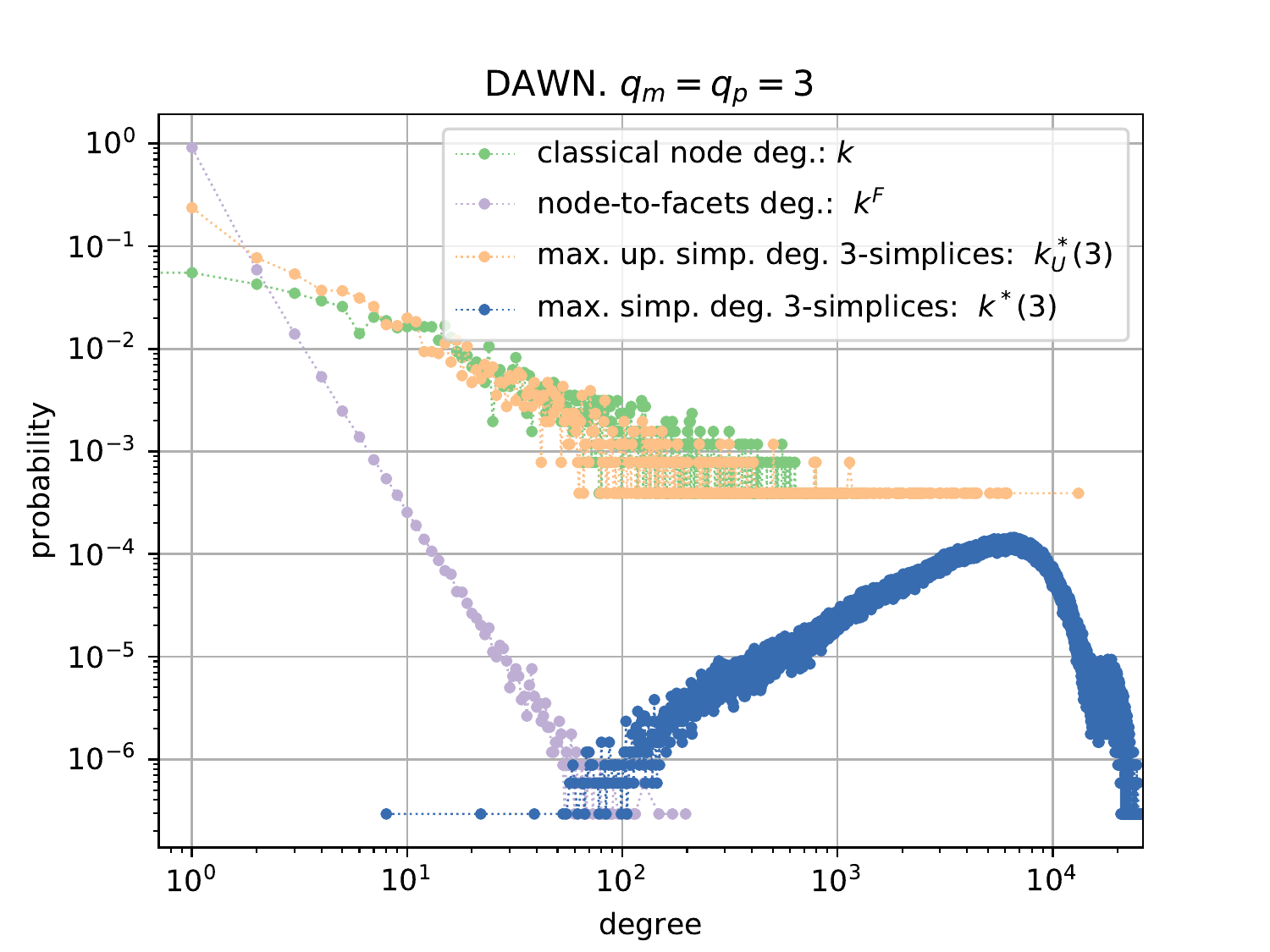}
	\caption{Log-log plot of simplicial degree distributions of tags-ask-ubuntu, tags-maths-sx and DAWN datasets.}
	\label{fig:G4}
\end{figure} 

In Figure \ref{fig:G4} (DAWN, tags-ask-ubuntu and tags-math-sx datasets) it is remarkable that the maximal simplicial $k^*$-degree distribution is surprisingly different to the classical node $k$-degree distribution. While the later reflects a sublinear preferential attachment regime, the maximal simplicial degree distribution exhibits, out of the high $k^*$-degree region, a bell-shaped Poisson-like curve typical to that of a random network, which in the high $k^*$-degree region (the tail cutoff) transforms into a power law with exponential cutoff, a graphical representation much closer to the classical node degree distribution (but with faster decay). We believe that this rare phenomenon has certain consistency since, as in the classical bell-shaped random models, the peak of the bell-curve for the maximal simplicial degree distribution is also achieved for the average maximal simplicial degree, $\langle k^*(q_m)\rangle$, of the datasets (fact that can be contrasted using Table \ref{tab:maxsimpdeg} and Figure \ref{fig:G4}).

\begin{figure}[htb]
	\centering
	\includegraphics[width=0.48\linewidth]{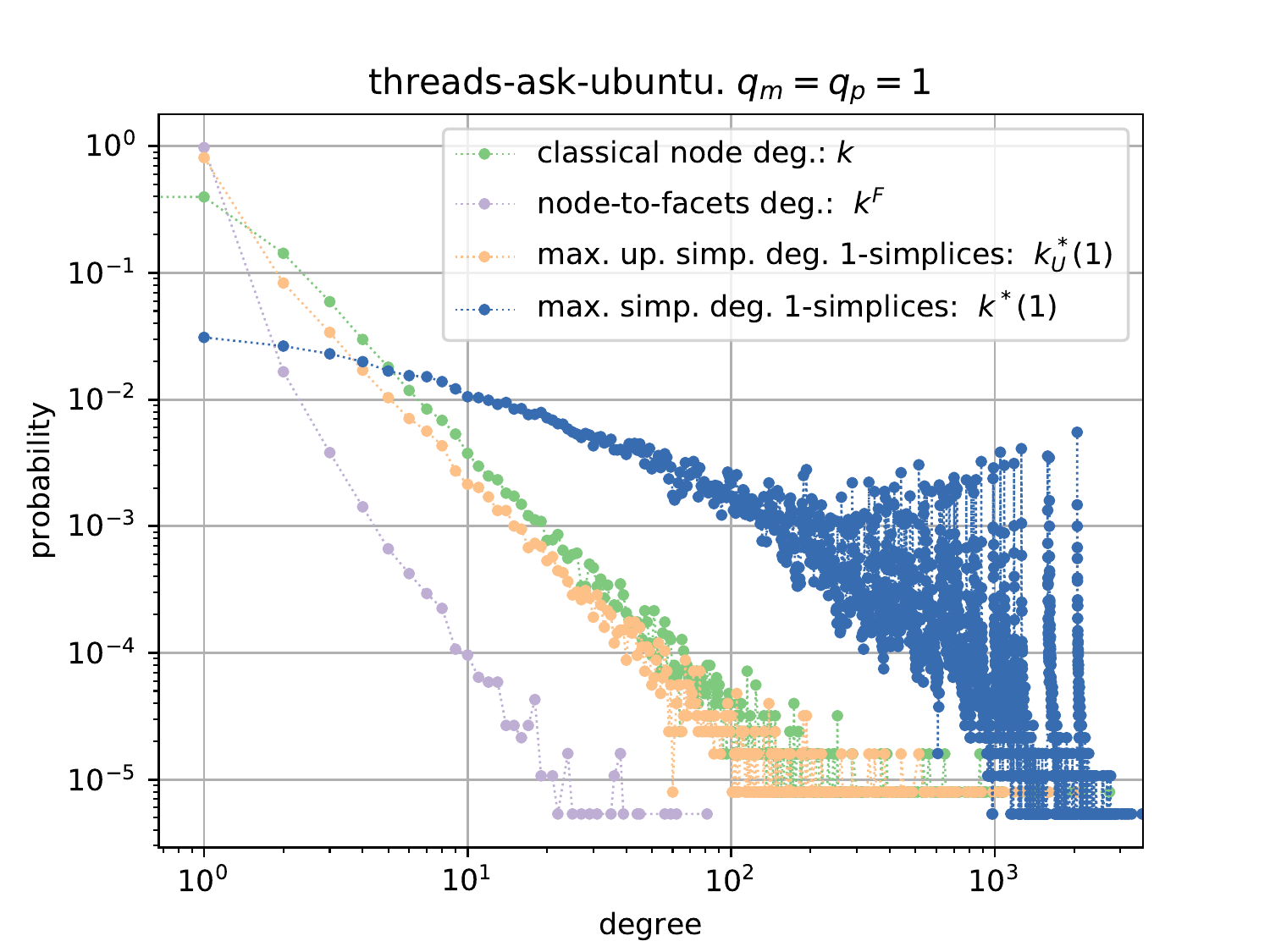}
	\includegraphics[width=0.48\linewidth]{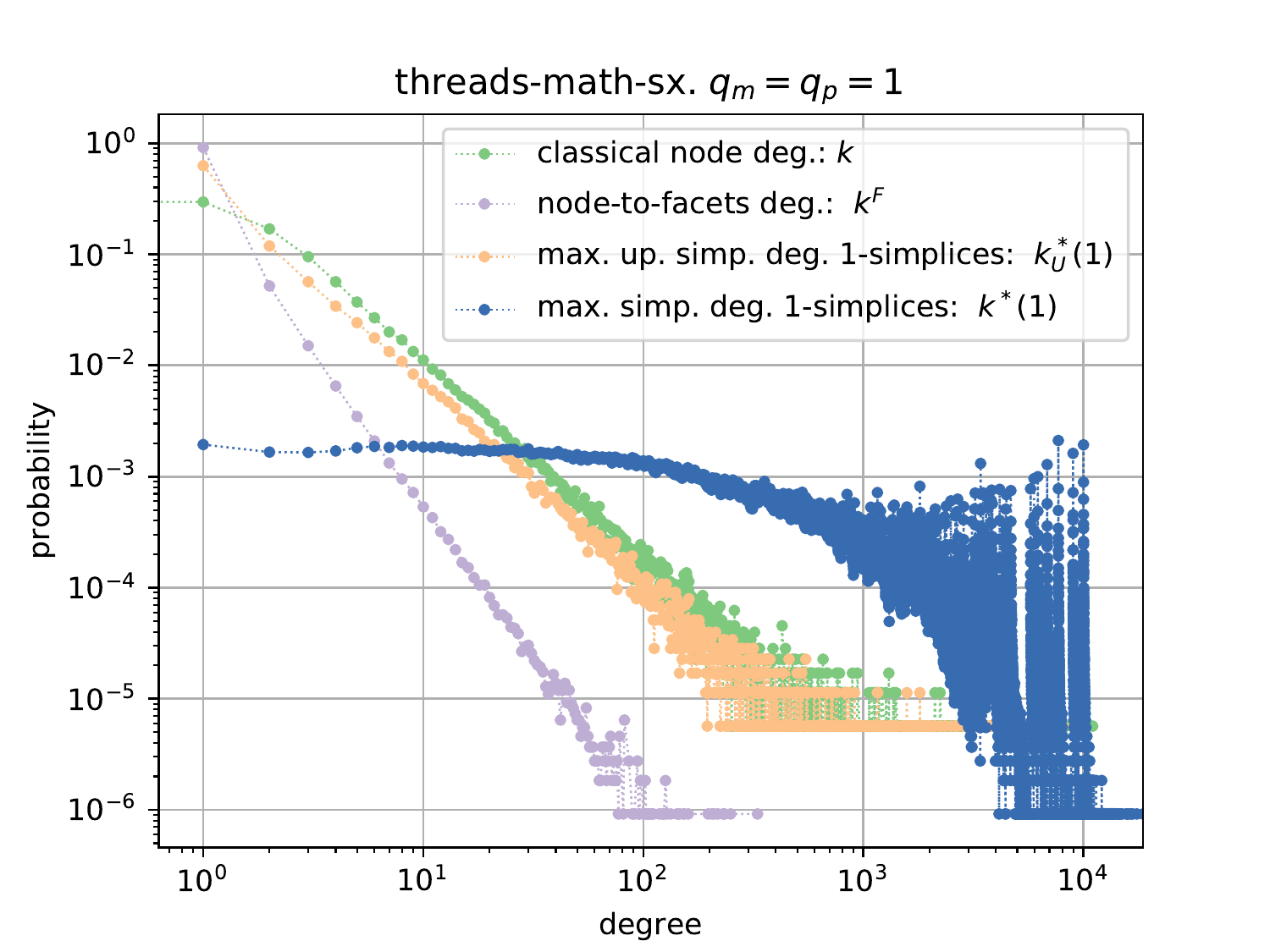}
	\caption{Log-log plot of simplicial degree distributions of threads-ask-ubuntu and threads-math-sx datasets.}
	\label{fig:G5}
\end{figure}

Figure \ref{fig:G5}  (threads-ask-ubuntu and threads-math-sx datasets) still needs to be properly interpreted. The maximal upper simplicial $k_U^*$-degree distribution is showing a kind of superlinear preferential attachment regime, getting thus closer to a richer-gets-richer process. The maximal simplicial degree $k^*$-distribution seems to present a higher small $k^*$-degree saturation and a tail cutoff where the probabilities go up, which seems to lead to a winner-takes-all phenomenon.

In addition, it seems that these higher-order distribution models may depend on the dimension of the simplex we are considering: the bigger the simplex dimension is, the closer to a random model the distribution is.

\section{Conclusions and future research}\label{s:concl}

Many real networks in the social, biological and biomedical sciences or computer science have an inherent structure of simplicial complexes, which reflect the multi interactions among agents and groups of agents. As far as we know, higher-order notions of adjacency and degree for simplices valid for any dimensional simplicial comparison are lacked in the literature. We have proposed these notions, shown their combinatoric properties, and we have given explicit methods for computing them introducing a multi-parameter combinatorial Laplacian on a simplicial complex (which generalizes the knowns graph and combinatorial Laplacians). These notions of higher-order simplicial adjacencies have allowed us to define two important degrees for the potential applications: the maximal upper simplicial degree of a simplex (counting the number of distinct facets the simplex is nested in); and the maximal simplicial degree of a simplex (counting the number of distinct facets the simplex belongs to and also the different facets its strict faces are contained in). We have studied their associated degree distributions in many and diverse real-world datasets, and we have shown a rich variety of higher-order connectivity structures that, for datasets of the same type, reflect similar simplicial collaboration patterns. Furthermore, we have shown that if we use the maximal simplicial degree, then its associated degree distribution is, in general, surprisingly different from the classical and upper adjacency ones, and shows that a random component is being intrinsically added to the initial preferential attachment nature of many of the datasets (since in general the log-log plot of the maximal simplicial degree distribution present a higher small degree saturation to that of the classical node degree distribution). In addition, this study reveals the existence of ``simplicial hubs''  (an equivalent notion to that of an influencer in social networks which could be referred to as ``simplicial influencers'' or ``influential communities'') and shows that they represent a deep organizing principle in the simplicial network topology.

In subsection \ref{ss:spec}, we have proposed potential theoretical applications for the multi-combinatorial Laplacian, like the study of generalized random walks in simplicial complexes and the spectral properties of the multi-Laplacian operator, both aimed at understanding and calculating new topological invariants in simplicial networks. 

In \cite{HS19} we have introduced new centrality measures on simplicial complexes based on the theoretical notions proposed in these notes (such as a simplicial clustering coefficient or closeness and betweeness centralities using a new notion of simplicial distance). It would be interesting to do an statistical study of these centrality measures in simplicial complex systems in the future, since it would represent a starting point for studying the geometry and robustness of the simplicial networks, together with the importance of certain simplicial communities in how information flows through a simplicial network.

Let us point out that, in many systems, empirical data is not always enough to distinguish distributions, and that data collection incompleteness might affect to the tail cutoff of the degree distributions representations. Thus, in order to state a proper classification of simplicial networks or to study the dynamics of a simplicial network, generative simplicial models predicting the expected simplicial degree distributions from a rigorous analytically point of view are still needed. This study might help in constructing and understanding connectivity algorithms and configurations models in simplicial complexes, which would allow to generalize some of the results in \cite{BR02,BA16} to the simplicial case and go beyond the $d$-pure simplicial networks studied in \cite{BK19}. Hopefully, these notes contribute to developing new applications of Topological Data Analysis in complex systems, and thus to expand the basis of an emergent Simplicial Network Science.

\section{Appendix}\label{s:A}\quad 

We give in this section the proofs of Theorems \ref{t:2}, \ref{t:3} and \ref{t:4}.

\subsection{Proof of Theorem \ref{t:2}.}\quad 

Let $p,h$ and $h'$ be non negative integers, put $q=p+h,\, q'=p+h'$ and fix 
$\{\tau_1^{(q')},\dots,\tau_m^{(q')}\},\,\{\sigma_1^{(q)},\dots,\sigma_n^{(q)}\}$ and $\{\gamma_1^{(p)},\dots,\gamma_r^{(p)}\}$ basis of $C_{q'}(K),\, C_{q}(K)$ and $C_{p}(K)$ respectively. Given the boundary operators $\partial_{q,h}\colon C_q(K)\to C_p(K)$ and  $\partial_{q',h'}\colon C_{q'}(K)\to C_p(K)$, denote by $B_{q,h}$ and $B_{q',h'}$ their corresponding matrices with respect to those bases. For the composition 
$$\xymatrix{C_q(K)\ar[r]^{\partial_{q,h}} & C_p(K)\ar[r]^{\partial^*_{q',h'}} & C_{q'}(K)}$$ one has:
\begin{align*}
(\partial^*_{q',h'}\circ\partial_{q,h} )(\sigma_j^{q})&=\displaystyle\sum_{i,k} b_{ij}^{(q,h)}b_{ik}^{(q',h')}\tau_k^{(q')}\\
&=\displaystyle\sum_{i,k}\sig_L\big(\sigma_j^{(q)},\gamma_i^{(p)};\gamma_i^{(p)}\big)\sig_L\big(\tau_k^{(q')}, \gamma_i^{(p)};\gamma_i^{(p)}\big)\tau_k^{(q')}\\
&= \displaystyle\sum_{i,k}\sig_L\big(\sigma_j^{(q)},\tau_k^{(q')};\gamma_i^{(p)})\tau_k^{(q')}\,.
\end{align*}

$\gamma_i^{(p)}$ is a $p$-face of both $\sigma_j^{(q)}$ and $\tau_k^{(q')}$ if and only if $|\sig_L\big(\sigma_j^{(q)},\tau_k^{(q')};\gamma_i^{(p)})|=1$, so we obtain that:
\begin{itemize}
\item  the number of $p$-faces of both $\sigma_j^{(q)}$ and $\tau_k^{(q')}$ is $$\ord^p_L(\sigma_j^{(q)},\tau_k^{(q')})=\displaystyle\sum_{i}|\sig_L\big(\sigma_j^{(q)},\tau_k^{(q')};\gamma_i^{(p)})|=\displaystyle\sum_{i} |b_{ij}^{(q,h)}||b_{ik}^{(q',h')}|\,,$$

\item the number of $q'$-simplices which are $p$-lower adjacent to $\sigma_j^{(q)}$ in $\gamma_i^{(p)}$ is 
$$\displaystyle\sum_{k}|\sig_L\big(\sigma_j^{(q)},\tau_k^{(q')};\gamma_i^{(p)})|=\displaystyle\sum_{k} |b_{ij}^{(q,h)}||b_{ik}^{(q',h')}|.$$
\end{itemize}

We have to prove that (Theorem \ref{t:2}):
$$\deg_L^p(\sigma_j^{(q)})= -1+\sum_{q'=p}^{\dim K}\sum_k\min\big(1,\sum_{i}|b_{ij}^{(q,h)}||b^{(q',h')}_{ik}|\big)$$

Fixed basis as above, we define the following sign matrix:
\begin{equation}\label{eq:signsmatrix}
S_{q,h,h'}(j):=\begin{pmatrix}
|\sig_L\big(\sigma_j^{(q)},\tau_1^{(q')};\gamma_1^{(p)})|&\cdots&|\sig_L\big(\sigma_j^{(q)},\tau_1^{(q')};\gamma_r^{(p)})|\\
\vdots&&\vdots\\ 
\vdots&&\vdots\\ 
|\sig_L\big(\sigma_j^{(q)},\tau_m^{(q')};\gamma_1^{(p)})|&\cdots&|\sig_L\big(\sigma_j^{(q)},\tau_m^{(q')};\gamma_r^{(p)})|\end{pmatrix}
\end{equation}
whose $(k,i)$-th entry is: 

$$s^{(q,h,h')}_{ki}(j)=|b_{ij}^{(q,h)}||b_{ik}^{(q',h')}|=\begin{cases}1 & \text{ if } \gamma_i^{(p)}\subseteq\sigma_j^{(q)}\cap\tau_k^{(q')}\\ 0 & \text{otherwise}\end{cases}$$

Note that if $s^{(q,h,h')}_{ki}(j)\neq 0$ for some $i$, then $\sigma_j^{(q)}$ and $\tau_k^{(q')}$ are $p$-lower adjacent in order $\ord^p_L(\sigma_j^{(q)},\tau_k^{(q')})=\sum_{i}|b_{ij}^{(q,h)}||b^{(q',h')}_{ik}|$. Hence, the number of $q'$-simplices $p$-lower adjacent to $\sigma_j^{(q)}$, counted each one with its order, is  $\sum_{ki}|b_{ij}^{(q,h)}||b^{(q',h')}_{ik}|$. To avoid that $\tau_k^{(q')}$ to be counted more than once in $\deg_L^p(\sigma_j^{(q)})$, we consider the minimum between 1 and its order, so that (assuming $\tau_k^{(q')}\neq \sigma_j^{(q)}$):
$$
\tau_k^{(q')}\sim_{L_p} \sigma_j^{(q)}\iff \min\big(1,\sum_{i}|b_{ij}^{(q,h)}||b^{(q',h')}_{ik}|\big)=1\,.
$$
 Therefore, the number of $q'$-simplices $p$-lower adjacent to $\sigma_j^{(q)}$ is:
$$
\deg_L^{q-q',p}(\sigma_j^{(q)})=\begin{cases}
\displaystyle\sum_k\min\big(1,\sum_{i}|b_{ij}^{(q,h)}||b^{(q',h')}_{ik}|\big) & \text{ for } q'\neq q\\
\displaystyle\sum_k\big(\min\big(1,\sum_{i}|b_{ij}^{(q,h)}||b^{(q,h)}_{ik}|\big)\big)-1 & \text{ for } q'=q
\end{cases}
$$
and the result follows. 

Notice that if $q=q'$, then the degree $\deg_L^{0,p}(\sigma_j^{(q)})$ is $\displaystyle\sum_k\min\big(1,\displaystyle\sum_{i}|b_{ij}^{(q,h)}||b^{(q,h)}_{ik}|\big)$ minus 1 since, by definition, $\sigma_j^{(q)}$ is not $p$-lower adjacent to itself.

\subsection{Proof of Theorem \ref{t:3}.}\quad 

Taking into account the following notations and partial results, the proof of Theorem \ref{t:3} is analogous to $p$-lower degree case of Theorem \ref{t:2}.

Let $p,h$ and $h'$ be non negative integers, denote $q=p-h,\, q'=p-h'$ and fix 
$\{\tau_1^{(q')},\dots,\tau_m^{(q')}\},\,\{\sigma_1^{(q)},\dots,\sigma_n^{(q)}\}$ and $\{\gamma_1^{(p)},\dots,\gamma_r^{(p)}\}$ basis of $C_{q'}(K),\, C_{q}(K)$ and $C_{p}(K)$ respectively. We denote by $B_{q+h,h}$ and $B_{q'+h',h'}$ the corresponding matrices (with respect to those bases) to the boundary operators: $$C_p(K)\xrightarrow{\partial_{q+h,h}} C_{q}(K) \text{ and } C_{p}(K)\xrightarrow{\partial_{q'+h',h'}} C_{q'}(K)\,.$$

For the composition
$$C_q(K)\xrightarrow{\partial^*_{q+h,h}} C_{p}(K)\xrightarrow{\partial_{q'+h',h'}} C_{q'}(K)$$ one has: 
\begin{align*}
(\partial_{q'+h',h'}\circ\partial^*_{q+h,h} )(\sigma_j^{q})&=\displaystyle\sum_{i,k} b_{ji}^{(q+h,h)}b_{ki}^{(q'+h',h')}\tau_k^{(q')}\\
&= \displaystyle\sum_{i,k}\sig_U\big(\sigma_j^{(q)},\tau_k^{(q')};\gamma_i^{(p)})\tau_k^{(q')}\,.
\end{align*}

Since $\gamma_i^{(p)}$ is a $p$-simplex containing $\sigma_j^{(q)}$ and $\tau_k^{(q')}$ as faces if and only if $|\sig_U\big(\sigma_j^{(q)},\tau_k^{(q')};\gamma_i^{(p)})|=1$ we obtain that:
\begin{itemize}
\item  the number of $p$-simplices which contain both $\sigma_j^{(q)}$ and $\tau_k^{(q')}$ as faces is $$\ord^p_U(\sigma_j^{(q)},\tau_k^{(q')})=\displaystyle\sum_{i}|\sig_U\big(\sigma_j^{(q)},\tau_k^{(q')};\gamma_i^{(p)})|=\displaystyle\sum_{i} |b_{ji}^{(q+h,h)}||b_{ki}^{(q'+h',h')}|\,,$$

\item the number of $q'$-simplices $\tau_k^{(q')}$ such that $\sigma_j^{(q)}$ and $\tau_k^{(q')}$ are faces of $\gamma_i^{(p)}$ is 
$$\displaystyle\sum_{k}|\sig_U\big(\sigma_j^{(q)},\tau_k^{(q')};\gamma_i^{(p)})|=\displaystyle\sum_{k} |b_{ji}^{(q+h,h)}||b_{ki}^{(q'+h',h')}|\,,$$

\end{itemize}

\subsection{Proof of Theorem \ref{t:4}.}\quad 

Assume $\sigma^{(q)}$ is a $q$-simplex, the $p$-adjacency degree of $\sigma^{(q)}$ has been defined as the number of $q'$-simplices $\sigma^{(q')}$ such that 
$\sigma^{(q)}\sim_{L_{p^*}} \sigma^{(q')}$ and $\sigma^{(q)}\not\sim_{U_{p'}} \sigma^{(q')}$, with $p'=q+q'-p$, and its maximal $p$-adjacent degree is the number of $q'$-simplices $\sigma^{(q')}$ such that 
$\sigma^{(q')}\sim_{A_p} \sigma^{(q)}$ and $\sigma^{(q')}$ is not a face of a $q''$-simplex $\sigma^{(q'')}$ which is also $p$-adjacent to $\sigma^{(q)}$ (see Definitions \ref{d:qhAdj} and \ref{d:Adjdeg}). 
As we have already remarked, the fact that a $q'$-simplex $\sigma^{(q')}$ to be $p$-lower adjacent to $\sigma^{(q)}$ can be encoded in terms of the lower sign of both simplices. In other words, if one sets a base $\{\gamma_1^{(p)},\dots,\gamma_r^{(p)}\}$ of $C_{p}(K)$, then a $q'$-simplex $\sigma^{(q')}$ is $p$-lower adjacent to $\sigma^{(q)}$ if and only if $\sig_L(\sigma^{(q)},\sigma^{(q')};\gamma_i^{(p)})\neq 0$ for some $\gamma_i^{(p)}$, so that, one has:

\begin{align*}
\sigma^{(q')}\sim_{L_p}\sigma^{(q)}\iff &\displaystyle\sum_{i} |\sig_L\big(\sigma^{(q)},\sigma^{(q')};\gamma_i^{(p)})|\geq 1\\ \iff &\min\Big(1,\displaystyle\sum_{i} |\sig_L\big(\sigma^{(q)},\sigma^{(q')};\gamma_i^{(p)}|\big)\Big)=1\,.
\end{align*}
In a similar way, one has: 

\begin{align*}
\sigma^{(q')}\sim_{U_{p}}\sigma^{(q)}\iff &\displaystyle\sum_{i} |\sig_U\big(\sigma^{(q)},\sigma^{(q')};\gamma_i^{(p)})|\geq 1\\ \iff &\min\Big(1,\displaystyle\sum_{i} |\sig_U\big(\sigma^{(q)},\sigma^{(q')};\gamma_i^{(p)}|\big)\Big)=1\,.
\end{align*}
 
For simplicity we shall denote:
\begin{align*}
m_L^p(\sigma^{(q)},\sigma^{(q')})=&\min\Big(1,\displaystyle\sum_{i} |\sig_L\big(\sigma^{(q)},\sigma^{(q')};\gamma_i^{(p)}|\big)\Big)\\ 
m_U^{p}(\sigma^{(q)},\sigma^{(q')})=&\min\Big(1,\displaystyle\sum_{i} |\sig_U\big(\sigma^{(q)},\sigma^{(q')};\gamma_i^{(p)}|\big)\Big)\\
adj^p(\sigma^{(q)},\sigma^{(q')})= &\, m^p_L(\sigma^{(q)},\sigma^{(q')})\big(1-m^{p+1}_L(\sigma^{(q)},\sigma^{(q')})\big)\big(1-m^{p'}_U(\sigma^{(q)},\sigma^{(q')})\big)\,. 
\end{align*}
Hence, 

\begin{equation}\label{eq:adjacency}
\begin{aligned}
\sigma^{(q')}\sim_{A_p}\sigma^{(q)}\iff adj^p(\sigma^{(q)},\sigma^{(q')})=1\,.
\end{aligned}
\end{equation}

The proof of (1) in Theorem \ref{t:4} follows from formula (\ref{eq:adjacency}). To prove (2) we need only to check the number of those $q'$-simplices, which being  $p$-adjacent to $\sigma_j^{(q)}$, are also faces of $q''$-simplices adjacent to $\sigma_j^{(q)}$. Assume that $\sigma_k^{(q')}$ is a $q'$-simplex such that the following hold:
\begin{itemize}
\item [(a)] $\sigma_k^{(q')}\sim_{A_p}\sigma_j^{(q)}$
\item [(b)] There exists a (likely not uniquely determined) $q''$-simplex $\sigma^{(q'')}$ such that $\sigma_k^{(q')}\subset\sigma^{(q'')}$ and $\sigma^{(q'')}\sim_{A_p}\sigma_j^{(q)}$.
\end{itemize}
By assumption $(a)$, we have that $adj^p(\sigma_j^{(q)},\sigma_k^{(q')})=1$, and assumption $(b)$ is equivalent to say that $|\sig_L(\sigma^{(q')}_k,\sigma^{(q'')};\sigma_k^{(q')})|\cdot adj^p(\sigma_j^{(q)},\sigma^{(q'')})=1$. Therefore, under these assumptions, $\sigma_k^{(q')}$ is not a maximal $p$-adjacent to $\sigma_j^{(q)}$ simplex, and we don't have to take this simplex into account to compute $\deg^{p^*}_A(\sigma_j^{(q)})$. Thus, for every $\sigma_k^{(q')}$, the expression: 
$$\displaystyle\sum_{q''=q'+1}^{\dim K}\displaystyle\sum_{\ell=1}^{\dim C_{q''}(K)}|\sig_L(\sigma^{(q')}_k,\sigma_\ell^{(q'')};\sigma_k^{(q')})|\cdot adj^p(\sigma_j^{(q)},\sigma_\ell^{(q'')})$$ gives the number of $q''$-simplices (where $q''$ runs over all dimensions from $q'+1$ to $\dim K$) which are $p$-adjacent to $\sigma_j^{(q)}$ and contain $\sigma_k^{(q')}$. Then $\sigma_k^{(q')}$ is $p$-adjacent but no maximal to $\sigma_j^{(q)}$ if and only if: 
$$\min\Big(1,\sum_{q''}\sum_{\ell}|\sig_L(\sigma^{(q')}_k,\sigma_\ell^{(q'')};\sigma_k^{(q')})|\cdot adj^p(\sigma_j^{(q)},\sigma_\ell^{(q'')})\Big)\cdot adj^p(\sigma_j^{(q)},\sigma_k^{(q')})=1\,,$$ and we get the statement. 
\vfill

\end{document}